\documentclass[letter,singlecolumn]{scrartcl}

\usepackage[T1]{fontenc}

\usepackage{graphicx}
\usepackage{amssymb}  
\usepackage{xcolor,mathrsfs}
\usepackage{%
	algorithm,%
    algpseudocode,
	amsthm,%
	bbm,%
	enumerate,%
	hyperref,%
	mathtools,%
	paralist,%
	pgf,%
	pgfplots,%
	tikz,%
	tikzscale,%
}
\usepackage{verbatim}
\usepackage{pgfplots}
\usepgfplotslibrary{fillbetween}
\usetikzlibrary{patterns}
\usepackage{subcaption}
\newcommand{\uv}{\textbf{u}}
\newcommand{\w}{\textbf{w}}
\newcommand{\x}{\textbf{x}}
\newcommand{\y}{\textbf{y}}

\newcommand{\q}{\textbf{q}}

\newcommand{\e}{\lambda}
\newcommand{\J}{\mathcal{J}}
\newcommand{\G}{\mathcal{G}}

\newcommand{\done}{\hfill$\square$}

\newtheorem{thm}{Theorem}
\newtheorem{prop}[thm]{Proposition}
\newtheorem{cor}[thm]{Corollary}

\newtheorem{lem}[thm]{Lemma}
\newtheorem{asu}{Assumption}
\newtheorem{defn}{Definition}

\newtheorem{exmp}{Example}

\begin{document}

\title{Learning Equilibria in  Coordination Games via Minorization-Maximization\footnote{This work was supported by D\'efi Inria-EDF}}  
\date{}
\author{Ashok Krishnan K.S.\footnote{Inria and DI ENS, Paris, France. \emph{Email:~ashok-krishnan.komalan-sindhu@inria.fr}}, H\'el\`ene {Le Cadre}\footnote{Inria, Lille, France. \emph{Email:~helene.le-cadre@inria.fr}}, Ana Bu\v si\'c\footnote{Inria and DI ENS, Paris, France. \emph{Email:~ana.busic@inria.fr}}}

\maketitle
\begin{abstract}
This paper considers games where the utilities for agents are the sum of a term proportional to a social utility, and another term that is an  individual cost or reward. The agents are assumed to be irrational in their perception of the individual cost or reward. The multi equilibrium game is regularized, and its strictly concave potential function is used to select a unique equilibrium. This selected equilibrium is shown to be an $\epsilon-$equilibrium of the original game, where $\epsilon$ is parametrized by the regularizing function. A minorization-maximization based iterative learning scheme is proposed to learn equilibria in this game. This scheme converges to the potential-optimal equilibrium, and has superior convergence behavior in comparison to gradient and best response methods.
\end{abstract}

\section{Introduction}\label{}
\noindent In many games and optimization problems, an agent’s payoff combines an individual cost with a reward proportional to a collective social utility. This captures scenarios where agents independently pursue a shared objective, making success dependent on coordination. Such interactions naturally form coordination games~\cite{cooper1999coordination} and, in some cases, aggregative games, where each agent responds to an aggregate of others’ actions~\cite{jensen2010aggregative}. Under mild conditions, they admit a (weighted) potential game structure, enabling equilibrium analysis via a global function~\cite{voorneveld1999potential}. Examples for scenarios where such problems arise, include energy markets, where electricity consumers utilize a shared resource while collectively aiming to meet a common carbon emission target~\cite{ashok2025achieving}, and the multi agent exploration problem~\cite{hao2023exploration}, where multiple agents try to reach a common target or destination, while moving in a coordinated but distributed fashion. A particular example of multi agent exploration is the rendezvous problem~\cite{cortes2017coordinated}, in which multiple robots, initially distributed across an unknown or unexplored environment, seek to reach a common destination while incurring individual energy costs. In all these cases, performance ultimately depends on the agents’ ability to coordinate through the effects of their collective actions.

While not all potential games correspond to meaningful social utilities, the closely related class of social purpose games~\cite{gilles2023emergent} explicitly decomposes utilities into a collective-benefit term and an individual cost, capturing a broad range of problems involving common or public resources.

From another perspective, the presence of a collective objective connects this framework to cooperative game theory~\cite{mccain2008cooperative}, where one studies how groups of agents generate and share value. In particular, when the underlying social utility exhibits supermodularity, one can leverage strong structural properties—such as increasing returns—to obtain existence, monotonicity, and robustness guarantees for equilibria (see, e.g.,~\cite{bach2013learning}). Finally, in the large-population regime, when the impact of any single agent on the aggregate becomes negligible, the model relates to mean-field games~\cite{lasry2007mean}. This approximation is attractive for analyzing asymptotic system behavior, but comes at a cost: by averaging interactions, it tends to wash out individual heterogeneity and may fail to capture nuanced agent preferences or cost perceptions.

In this paper, we consider such an additive utility model. However, we assume that agents are not fully rational in their perception of individual costs. When agents are human, this deviation from rationality may arise from behavioral biases~\cite{jallais2005allais}. For non-human agents, bounded rationality may instead stem from limitations in information, perspective, or computational capacity, which can hinder effective coordination. Several modeling paradigms have been proposed to capture such deviations from perfect rationality. Quantal best-response~\cite{mckelvey1995quantal} and noisy discrete choice models~\cite{alos2010logit} introduce stochasticity in decision-making, typically assuming that agents respond smoothly to payoff differences (e.g., via logit dynamics). These models are analytically tractable and well-suited for equilibrium learning, but they primarily capture random errors rather than systematic biases in perception. Distributionally robust game~\cite{aghassi2006robust} formulations, on the other hand, account for ambiguity in agents’ beliefs by considering worst-case distributions over uncertainties. While they provide strong robustness guarantees, they tend to reflect ambiguity aversion at the population level rather than individual-level cognitive distortions, and may lead to overly conservative behaviors.

In contrast, we adopt prospect theory (PT)~\cite{kai1979prospect}, which explicitly models systematic deviations from expected utility maximization by incorporating reference dependence, loss aversion, and probability distortion. Prospect theory addresses well-documented shortcomings of classical von Neumann–Morgenstern expected utility theory (EUT)~\cite{von1944theory}, whose assumption of fully rational agents is often contradicted by experimental evidence~\cite{jallais2005allais}. A variety of prospect-theoretic models exist~\cite{stott2006cumulative}, differing in how they model value functions and probability weighting, and often tailored to specific applications.

Many games, including those with prospect-theoretic utilities~\cite{ks2025irrationality}, admit multiple equilibria, complicating both analysis and learning. To address this, we introduce a regularized (perturbed) game, adding a small term to the utility or potential function. The regularized potential function has a unique maximum, which selects a unique equilibrium while preserving the structure of the original game. This is analogous to regularization in machine learning to ensure unique maximizers. The regularizer thus acts as a weak equilibrium selection mechanism, linking the perturbed equilibrium to the original set of equilibria. Learning this single equilibrium is the next question of interest. The kind of algorithm that can be used for learning is intricately tied to the structure of the game. Monotone games and potential games show convergence under gradient type learning dynamics~\cite{rosen1965existence}. In the case of aggregative games, best response learning converges under strong assumptions on the structure of the interaction between the users' strategies and utilities~\cite{kukushkin2004best}. However, the performance of many learning algorithms worsens when the game is non smooth, i.e., the utilities are no longer continuously differentiable. To address  this challenge, we propose a minorization–maximization (MM)~\cite{hunter2004tutorial} based learning scheme. The  key idea is to replace the original  optimization problem with a sequence of surrogate problems whose optima converge to a chosen equilibrium. A carefully chosen surrogate can accelerate updates and speed up convergence to equilibrium. The required aggregate information is communicated through a coordinating agent, enabling scalable learning despite the presence of non-standard utility functions. Since MM can be used to learn equilibria in games with smooth and non-smooth utilities, we see that it presents itself as an algorithm which can be used across a variety of applications ranging from  electricity markets to multi agent coordination.

Utilities reflecting PT preferences encode sensitive information (e.g., risk attitudes, reference points). Our framework requires only aggregate information, protecting individual preferences. The surrogate optimization in MM further supports privacy-preserving modifications, such as adding noise or using secure aggregation, without compromising convergence to the perturbed equilibrium.
\section{Literature Review}
\noindent For games with compact strategy sets and continuous utilities, the existence of mixed strategy Nash equilibria is known~\cite{glicksberg1952further}. However, in order to establish the existence of pure Nash equilibria, further assumptions need to be made on the structure of the utility function. This leads to different classes of games, with different types of results on the existence, uniqueness and learnability of equilibria. The class of concave games~\cite{rosen1965existence} consists of games with concave utilities, and  pure strategies chosen from compact, convex sets. For concave games,  under the additional assumption that the utility functions are smooth, the uniqueness of Nash equilibrium follows from the strict monotonicity of the pseudogradient of the utility. The pseudogradient also gives a direction of descent to the Nash equilibrium. Generalizations of the monotonicity condition are available in \cite{facchinei2003finite}. Socially concave games~\cite{even2009convergence} are a subclass of concave games, where the utility of an agent varies as a convex function of the strategy of another agent. In this class of games, no-regret learning converges to Nash equilibria; in general, no-regret learning is known to converge only to correlated equilibria~\cite{hart2000simple}.

The class of supermodular games~\cite{milgrom1990rationalizability} have Nash equilibria under the assumptions that the strategy sets are lattices, and the utilities are twice differentiable with supermodular structure. While, in general, there is no guarantee of uniqueness, if the equilibrium is unique, the game is dominance solvable. In potential games~\cite{voorneveld1999potential}, there exists a potential function whose optima are Nash equilibria of the game. Thus, a game with a continuous potential over a compact strategy space, will have Nash equilibria. If the potential is concave and smooth, the set of maximizers of the potential function is identical to the set of Nash equilibria of the game~\cite{neyman1997correlated}. Learning the optima of the potential function recovers Nash equilibria of the game. Learning schemes for potential games with differentiable utilities are available in~\cite{monderer1996potential}. Convergence results for best- and better-response dynamics in near potential games, for games with finite strategy sets and mixed strategies, are provided in~\cite{candogan2011learning}. Aggregative games~\cite{jensen2018aggregative} are games where the utility of a user depends on the strategies of other users through an aggregating function. In \cite{jensen2010aggregative} it was shown that any (quasi) aggregative game where the best reply correspondence has a decreasing selection, and the aggregation function has the property of increasing or decreasing differences, has a Nash equilibrium.
 Social purpose games~\cite{gilles2023emergent} combine potential and aggregative properties. In~\cite{gilles2023emergent} the authors focus on the emergence of cooperative behavior in such games. Convergence of better-reply dynamics in games that are bothe concave and aggregative games is studied in~\cite{dindovs2006better}. It shows that there is convergence to Nash equilibrium under boundedness assumptions on the derivative of the best-response correspondence.

Prospect theory has been used in economics for many decades, in order to model decision making under uncertainty~\cite{holmes2011management}. Most works using PT in game theory have focused on the question of existence of equilibria. Some such works that use PT models in finite games include \cite{keskin2016equilibrium}, \cite{merrick2016modeling} and \cite{vahid2019modeling}. A notion of correlated equilibrium for finite games using PT is introduced in \cite{keskin2016equilibrium}, and its existence is shown. In~\cite{merrick2016modeling}, PT is used to model the non rational behavior of an attacker in a combative scenario. Using a value function proposed in~\cite{tversky1992advances}, they show that the best response of the defender depends on this value function. A bidding game between electricity generation companies for participation in an energy market is studied in~\cite{vahid2019modeling}. Here, PT is used for modelling the subjectivity of the managers of electricity companies, which plays a significant role in the bidding process. The model is used to show how subjectivity causes deviation from the expected rational behavior. It must also be noted that the behavior of the PT game is strongly influenced by the nature of the PT function used. In most papers, authors use a single form of the PT function, chosen either for analytical convenience or from empirical conclusions. Authors have used different forms of the PT function, including piecewise linear~\cite{shalev2000loss}, piecewise polynomial~\cite{metzger2019non} and others~\cite{stott2006cumulative}. Power systems economics modelling is another area where PT has been used widely; for a broad survey see~\cite{gan2022application}. The presence of human agents making decisions about energy purchase and utilization, makes it especially relevant to apply PT to model human subjectivity in the face of risk.

Using prospect theory in games with continuous strategy sets comes with technical challenges, since a number of  PT value functions are non-smooth and/or non-concave. In~\cite{fochesato2025noncooperative}, the authors propose the notion of critical Nash equilibrium in PT games with non-smooth and non-concave utilities. This is a weaker (local) notion of equilibrium, which is useful when a global equilibrium is not well defined. They provide conditions for the existence of a critical Nash equilibrium, and use the notion of Price of Anarchy to quantify the imapct of irrationality on agent behavior. In \cite{ks2025irrationality}, the effect of PT transforms on the set of equilibria is characterized analytically, and the results applied to an electricity market. They show existence of Nash equilibria under the assumption of concavity of the PT value function. The PT utility leads to multiple equilibria. While they are able to access some of the equilibria by a gradient based learning method, they do not provide a theoretical guarantee of convergence. In \cite{ashok2025achieving}, an incentive scheme is developed for driving agents to a desirable equilibrium, for a game with a separable PT based utility structure. The authors assume that the game is smooth and strictly monotone, which restricts the set of applicable PT transformations. The present paper generalizes \cite{ashok2025achieving} in several respects. First, it extends the mathematical framework to encompass nonsmooth models, without requiring strict monotonicity of the collective utility. Consequently, the model studied in \cite{ashok2025achieving} is recovered as a special case of our more general setting. Second, we propose a novel learning algorithm based on minorization-maximization (MM) and establish its convergence properties. Thirdly, we prove convergence of not only the MM, but also of the gradient and best-response approaches to learning, under different assumptions on the smoothness of the constituent functions of the utility. Finally, we revisit a variant of the illustrative example from \cite{ashok2025achieving} in Section \ref{sec:Numerics}, where numerical results demonstrate improved convergence behavior with the proposed algorithm.

\subsection{Contributions}

\noindent The key contributions of this paper are as follows.
\begin{enumerate}
    \item We develop a game model which involves agent utilities that combine collective (social) rewards with an individual reward/cost that agents perceive in a biased or irrational manner. Agents need to satisfy social as well as individual requirements. This models a wide variety of applications that involve coordination between agents while also addressing individual bounded rationality. We use prospect theory to model agent irrationality arising due to variable reward perception under risk. 
    \item Using the potential structure of the regularized game, we obtain a unique equilibrium which is also the maximum of the potential function associated with the regularized game. Thus, regularization and potential formulation can be viewed together as an equilibrium selection method. We show that, for a game where the collective benefit is differentiable, the set of maximizers of the potential function is identical to the set of Nash equilibria of the game. We characterize theoretically how the selected equilibrium relates to the set of equilibria of the unregularized game, in terms of the regularizing function.

    \item In addition to gradient and best-response dynamics, we propose a minorization-maximization based learning scheme to learn the equilibrium of the game. Both (sub)gradient and best-response learning have limitations when the game is non smooth. Subgradient methods can be generally slower, while best response does not always converge to the potential optimal equilibrium. The minorization-maximization based learning scheme, on the other hand, is shown to converge to the potential-optimal equilibrium for non smooth utilities. It also converges exponentially fast.
\end{enumerate}
\subsection{Structure of the paper}
\noindent The rest of the paper is structured as follows. In Section \ref{sec:prospect}, we provide a short overview of prospect theory, followed by a summary of the notation. Section \ref{sec:Game Model} introduces the game theoretic model. Section \ref{sec:Existence} provides results on the existence of Nash equilibria. Section \ref{sec:characterize} relates the equilibria of the regularized and non regularized games. In Section \ref{sec:Learning}, we present different learning approaches to the game, and in Section \ref{sec:convergence}, we prove their convergence properties under different assumptions on smoothness. This is followed by numerical examples of the behavior of the discussed algorithms in Section \ref{sec:Numerics}, followed by the conclusion.

\section{Overview of Prospect Theory}\label{sec:prospect}
\noindent In this section, we explain the PT framework, using ideas from works such as 
 \cite{kai1979prospect} and \cite{stott2006cumulative}. 
 An agent obtains random rewards $(R_1,...,R_M)$ with probabilities $(q_1,...,q_M)$ where $q_1+...+q_M=1$. We call $(R_1,q_1,...,R_M,q_M)$ a \emph{prospect}. Assume that the rewards are ordered in increasing order of attractiveness. Now consider an agent having to choose between two prospects, $P:=(R_1,q_1,...,R_M,q_M)$ and $\hat P:=(\hat R_1,\hat q_1,...,\hat R_M,\hat q_M)$. In the classical EUT setting, this choice is made by comparing the expected utilities of the prospects, 
 \begin{align}
 \sum_{j=1}^Mq_jR_j \lessgtr \sum_{j=1}^M\hat q_j\hat R_j
 \end{align}
 and choosing the greater one.
 
 In prospect theory, the decision is made as follows. Each reward is viewed with respect to a (psychological) reference value. A reward which is greater than this reference is perceived as a gain, and a reward below this reference is understood as a loss.  When faced with gains, agents will behave in a \emph{risk averse} manner, and when faced with losses, agents tend to be \emph{risk seeking} or \emph{risk neutral}~\cite{kai1979prospect}. 
  This perception is captured by a value  function $V:\mathbb R\to\mathbb R$, as in Fig. \ref{fig:distortion-curve-val}. Here the agents are risk averse for gains and risk neutral for losses, with the reward zero representing the reference value that demarcates gains from losses.
  
\begin{figure}
    \centering
    \begin{tikzpicture}[scale=0.8]
\begin{axis}[
          xmax=5,ymax=4,
          xmin=-2,
          axis lines=middle,ticks=none]
\addplot[smooth,black,mark=none,
line width=1.5pt,domain=0:9.5,
samples=63]  {(1/log2(2.712))*log2(1+x)};
\addplot[smooth,black,mark=none,
line width=1.5pt,domain=-2:0,
samples=63]  {x};
\addplot[dotted,black,mark=none,
line width=1.5pt,domain=0:6,
samples=63]  {x};
\node at (axis cs:3,1.6) {$V$};
\node at (axis cs:4,-0.2) {reward};
\node at (axis cs:0,3) {perceived reward};
\end{axis}
\end{tikzpicture}
    \caption{Example of a prospect theoretic value function $V$ that maps rewards to perceptions. The dotted line is the unit slope line for reference.}
    \label{fig:distortion-curve-val}
\end{figure}
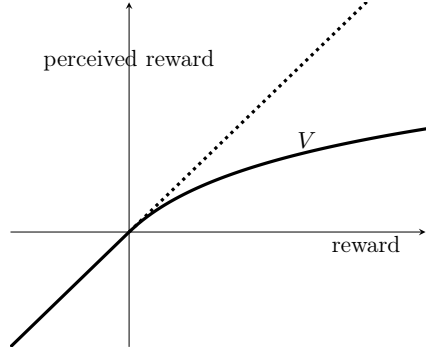

In addition to $V$, the probabilistic weight accorded to rewards is modified using a function $\pi$. The value of the prospect under PT is,
\begin{equation}\label{eqn:modified-value-PT}
    V=\sum_{j=1}^M \tilde q_j V(R_j),
\end{equation}
where $\tilde q_i$ is given by~\cite{stott2006cumulative},
\begin{align}
    \tilde q_1 &=\pi(q_1),\label{eq:q-tilde-gen-1}\\
    \tilde q_j &= \pi(\sum_{m=1}^j q_m)-\pi(\sum_{m=1}^{j-1} q_m),~j=2,...,M,\label{eq:q-tilde-gen-2}
\end{align}
for some monotone increasing $\pi:[0,1]\to[0,1]$. The map $\pi$ achieves the overweighting of small probabilities, and the underweighting of large probabilities, for large outcomes. Faced with a choice between two prospects, the agent compares their PT  values \eqref{eqn:modified-value-PT}, and chooses the better. Note that with $V(x)=x$ and $\pi(x)=x$, we retrieve the standard formulation of choice under the maximization of expected utilities.

\subsection{Notation and symbols}
\noindent The set of real numbers is $\mathbb R$. For a vector $\x:=(x_1,...,x_K)$, $\x_{-i}$ denotes the same vector with the $i$th element removed, i.e., $(x_1,..,x_{i-1},x_{i+1},...,x_K)$. If a vector $\x$ is a game theoretic strategy vector then we use it interchangeably with the game theoretic notation $\x:=(x_i,\x_{-i})$. For a random variable $X$ with distribution $d$, $\mathbb{E}_d [X]$ denotes its expectation under the distribution $d$. The indicator function is denoted by $\textbf{1}$. For a function $f$, we use $\partial f(\x)$ to denote its subdifferential (i.e., the set of its subgradients) at $\x$. If $f$ is differentiable at $\x$, its gradient is denoted by $\nabla f(\x)$.

\section{Game Model}\label{sec:Game Model}
\noindent Let $\mathcal N:=\{1,2,...,N+1\}$ be the set of players. Player $N+1$ acts solely as a coordinator between players $1,2,...,N$. For $i=1,..,N$, player $i$ chooses its strategy $x_i$ from a convex and compact set $\mathcal X_i\subset \mathbb R$. Let $\mathcal X:=\prod_{j=1}^N \mathcal X_j$ denote the set of joint strategies. The randomness in outcomes is modelled by a real valued random variable $\xi$ taking values over the set  $\Xi:=\{\xi_1,...,\xi_M\}$\footnote{Most of the results in this work can be extended to the case where $\xi$ has continuous distribution, albeit with more conditions on other variables to ensure integrability. We restrict ourselves to a finite support for simplicity.} where $0<\xi_1<\cdots <\xi_M$, with probability distribution $\q=(q_1,...,q_M)$. Let $V_i:\mathbb R\to\mathbb R,i=1,...,N$ and $\pi:[0,1]\to[0,1]$ be monotone increasing. The utility for user $i$  under the strategy $\x=(x_i,\x_{-i})$ where  $\x_{-i}:=(x_1,..,x_{i-1},x_{i+1},...,x_N)$, is given by
	\begin{align}\label{eq:utility-fn-def}
		J_i^{\e}(x_i,\x_{-i}) \triangleq a_i\J^{\e}(\x)+\mathbb{E}_{\tilde\q} [V_i\circ\mathcal R_i(x_i,\xi)],
	\end{align}
	where $0<a_N\le a_{N-1}\le ...\le a_1$, 
	and $\tilde\q$ is the distribution generated from $\q$ by \eqref{eq:q-tilde-gen-1}-\eqref{eq:q-tilde-gen-2} using $\pi$. This form of the utility function includes the following terms:
	\begin{enumerate}
		\item A regularized collective   usage benefit $\J^{\e}(\x)=\J(\x)+\e H(\x)$, where $\J:\mathcal X\to\mathbb R$ is a  collective utility common to all agents, $H:\mathbb R^N\to(-\infty,0]$ is a negative regularizing function, and $\e\ge 0$. 
		\item An individual random reward $\mathcal R_i(x_i,\xi):\mathcal X_i\times\Xi\to\mathbb R$. This reward is random due to its dependence on $\xi$, and is viewed subjectively through the PT value function $V_i$.
	\end{enumerate}

    The cost structure emphasizes a trade-off between collective  and individual benefit. This cost structure is commonly seen in mean field~\cite{lasry2007mean} and large population games~\cite{ma2011decentralized}, where the large population size causes each agent to view the actions of all its adversaries aggregatively. However, such a collective plus individual utility structure is seen in smaller population games as well. Indeed, games  where agents share a common resource can easily have this utility form. Routing games, where multiple agents seek to route traffic through a common network, are a good example. In the routing game in~\cite{chen2008altruism}, each agent's cost function is a convex combination of their individual latency, and the total latency experienced by all users. Closely related to routing games are congestion games~\cite{rosenthal1973class}, where  multiple agents using a common resource, such as a network or a road. The congestion experienced by each user is a function of the joint strategy of all users. The class of social purpose games~\cite{gilles2023emergent} uses a similar cost model to study resource sharing and other economic games. Games that deal with provisioning of public goods over a network~\cite{bramoulle2007public} also use a collective plus individual cost structure, where the collective term is a result of the network interaction. A classical economic model featuring a combination of collective and individual components in agents’ utilities is the problem of the commons~\cite{gibbons1992game}. In this setting, agents benefit from exploiting a shared resource, such as common grazing land for their goats. Each agent has an incentive to increase the size of their herd, but as all agents do so, the shared resource is depleted more rapidly, creating a conflict between individual incentives and collective welfare. Another example is the moral hazard problem in teams~\cite{holmstrom1982moral}, in which agents contribute to a joint effort but their individual efforts are not directly observable. Each agent’s utility combines a reward that depends on the collective outcome with an individual cost associated with their own effort. Building on this collective–individual utility structure, we investigate how agents' potentially irrational perceptions of individual rewards and costs affect strategic behavior and equilibrium outcomes.

    We make the following assumptions on the constituents of the utility function.
    \begin{asu}\label{assumption:contin}
        The functions $\J$, $\{V_i\}_{i=1}^N$, $\{\mathcal R_i\}_{i=1}^N$ and $ H$ are continuous in each of their respective arguments.
    \end{asu}
    \begin{asu}\label{assumption:compact}
        The sets $\{\mathcal X_i\}_{i=1}^N$ are convex and compact.
        \end{asu}
 Assumptions \ref{assumption:contin}-\ref{assumption:compact}, combined with the structure of the utility, leads to the existence of Nash equilibria for the game, as will be seen in Section \ref{sec:Existence}.  We define the following term before making the next assumption.
\begin{defn}
    Let $\mathcal Y\subseteq \mathbb R^N$. A function $\mathscr F:\mathcal Y\to\mathbb R$ is called strongly concave with modulus $\rho>0$, if  \begin{align}\label{eq:def-str-cvx}
    \mathscr F(\y_1)-\mathscr F(\y_2)\le \langle\y_1-\y_2,\w\rangle-\frac{\rho}{2}||\y_1-\y_2||^2,~\forall \w\in\partial\mathscr F(\y_2),\end{align}
where $\partial\mathscr F(\y_2)$ is the set of subgradients  of the function $\mathscr F(\cdot)$ at $\y_2$.
\end{defn}
    \begin{asu}\label{assumption:concave}
        The functions $\J, \{V_i\circ \mathcal R_i\}_{i=1}^N$ are concave in  their respective arguments, while $H$ is non negative, differentiable and strongly concave with modulus $\mu>0$.
    \end{asu}
    The assumption of concavity on the constituent functions of the utility is necessary to show convergence properties of learning algorithms. Strong concavity and differentiability of the regularizer improve the learning performance and lead to stronger results.
    \begin{defn}
        For $\e\ge0$, we define the $\e-$Game, $\G^{\e}$ as the tuple $\left(\mathcal N,\{\mathcal X_j\}_{j=1}^{N},\{J_j^{\e}\}_{j=1}^{N},\{V_j\}_{j=1}^N,\pi\right)$.
    \end{defn}    
Note that the game $\G^0$ is not regularized. For any $\e>0$, the game $\G^{\e}$ is regularized.  
   
\begin{defn}
    An $\epsilon$-Nash equilibrium for $\G^{\e}$ is any $\x^*:=(x^*_i,\x_{-i}^*)=(x_1^*,...,x_N^*)\in\mathcal X$ such that for all $i=1,...,N$, 
	\begin{align}\label{eq:def-nash-eq}
		J_i^{\e}(x^*_i,\x^*_{-i})\ge J_i^{\e}(x,\x^*_{-i})-\epsilon,~\forall x\in\mathcal X_i,
	\end{align}
    for some non negative $\epsilon$. 
\end{defn}
\begin{defn}
    A Nash equilibrium is an $\epsilon-$Nash equilibrium with $\epsilon=0$.  The set of Nash equilibria of the game is denoted by $\mathbf N^{\e}$.
\end{defn}
\begin{defn}
    The game $\G^{\e}$ is called \emph{smooth} if the utilities $J_i^{\e}$ are continuously differentiable in $x_i$, for each $i$.
\end{defn}

\section{Existence of Nash Equilibria}\label{sec:Existence}
\noindent It is easy to see that the game $\G^{\e}$ is a weighted potential game~\cite{monderer1996potential}.
\begin{lem}
    The game $\G^{\e}$ is a weighted potential game, with potential function
    \begin{align}
        \Phi^{\e}(\x)\triangleq 
        \mathcal J^{\e}(\x)+\sum_{j=1}^N\frac{1}{a_j}\mathbb{E}_{\tilde\q}\left[V_j\circ \mathcal R_j(x_j,\xi)\right].
    \end{align}
\end{lem}
\begin{proof}
    For all $i$, observe that for $\x=(x_1,...,x_N)=(x_i,\x_{-i})$ and $\tilde\x=(\tilde x_i,\x_{-i})$, we have
    \begin{align}\label{eq:pot-equal-util}
        J_i^{\e}(x_i,\x_{-i})-J_i^{\e}(\tilde x_i,\x_{-i})=a_i(\Phi^{\e}(\x)-\Phi^{\e}(\tilde \x)).
    \end{align}\done
\end{proof}
Thus, increasing the utility of any player corresponds to an increase in the potential function. Define the maximum of the potential function and its set of maximizers,
\begin{align}
\Phi^{\e\star} &\triangleq\max_{\x\in\mathcal X}\Phi^{\e}(\x),\\
\label{def:gamma-set}
\Gamma^{\e} &\triangleq \arg_{\x\in\mathcal X}\max \Phi^{\e}(\x).
\end{align}
All maximizers of the potential function are Nash equilibria.
\begin{lem}\label{lem:gamma-properties}
    The set $\Gamma^{\e}$ satisfies the following properties.
    \begin{enumerate}
        \item $\Gamma^{\e}$ is non empty, convex and compact.
        \item $\Gamma^{\e}\subseteq\mathbf{N}^{\e}$.
        \item for $\e>0$, $|\Gamma^{\e}|=1$.
    \end{enumerate}
\end{lem}
\begin{proof}
    Since $\Phi^{\e}$ is continuous (Assumption \ref{assumption:contin}) and $\mathcal X$ is compact (Assumption \ref{assumption:compact}), it follows by the Weierstrass theorem~\cite[Proposition 3.2.1]{bertsekas2009convex} that $\Phi^{\e}$ attains its maximum, and hence $\Gamma^{\e}$ is non empty and compact. From Assumption \ref{assumption:concave}, it follows that $\Gamma^{\e}$ is a convex set for all $\e$.
    Let $\x:=(x_i,\x_{-i})\in\Gamma^{\e}$. Let $\tilde\x:=(\tilde x_i,\x_{-i})$. Then
    \begin{align*}
        J_i^{\e}(x_i,\x_{-i})-J_i^{\e}(\tilde x_i,\x_{-i})=a_i(\Phi^{\e}(\x)-\Phi^{\e}(\tilde\x))\ge 0,
    \end{align*}
    using \eqref{eq:pot-equal-util} and since  $\x\in\Gamma^{\e}$. Since $\tilde x_i$ and $i$ are arbitrary, it follows that $\x\in\mathbf{N}^{\e}$. 
      When $\e>0$, we also see that $\Phi^{\e}$ is strictly concave, and therefore $|\Gamma^{\e}|=1$.

\done
\end{proof}
From the fact that $\Gamma^{\e}$ is non empty, we immediately have the following  result.
\begin{cor}\label{thm:existence}
     The game $\G^{\e}$ admits Nash equilibria.
\end{cor}

Any element of $\Gamma^{\e}$ is a Nash equilibrium as well as on optimizer of the potential function. 
\begin{defn}
    Any $\x^{\e}\in\Gamma^{\e}$ is called a potential-optimal equilibrium of the game $\G^{\e}$.
\end{defn}
In general, the set $\Gamma^{\e}$ is not identical to the set of Nash equilibria, $\mathbf{N}^{\e}$. However, if the collective utility function $\J$ is differentiable, we have the following equivalence.

\begin{lem}\label{lem:separable-nash-potential-connex}
If the collective utility function $\J$ is differentiable, all Nash equilibria of the game $\G^{\e}$ are potential optimal, and hence,
$\Gamma^{\e}=\mathbf N^{\e}.$
\end{lem}
\begin{proof}
    Since we  know from Lemma \ref{lem:gamma-properties} that $\Gamma^{\e}\subseteq \mathbf N^{\e}$, we proceed to prove inclusion in the other direction. Let $\hat\x\in\mathbf{N}^{\e}$. Then we know that
    \begin{align}
        J_i^{\e}(\hat\x)\ge J_i^{\e}(x,\hat\x_{-i})~\forall x\in\mathcal X_i.
    \end{align}
    Let us define $G_i(x)\triangleq\mathbb{E}_{\tilde\q} [V_i\circ\mathcal R_i(x,\xi)]$. Using this notation, we have, for all $i=1,...,N,$
    \begin{align}
        J_i^{\e}(\hat\x)\ge a_i\J^{\e}(\x)+G_i(x_i)~\forall~x\in\mathcal X_i.
    \end{align}
For $x\in\mathcal X_i$, and $\vartheta\in[0,1]$, define $\overline\x:=(1-\vartheta)\hat\x+\vartheta x\uv_i$, where $\uv_i$ is the vector that is zero at all coordinates except at the $i-$th, where it is one. Clearly $\overline\x\in\mathcal X$. We therefore have,
\begin{align}
    J_i^{\e}(\hat\x) &\ge a_i\J^{\e}(\overline\x)+G_i(\overline x_i),\\
    &=a_i\J^{\e}(\hat\x+\vartheta (x\uv_i-\hat\x))+G_i((1-\vartheta)\hat x_i+\vartheta x),\\
    &=a_i\J^{\e}(\hat\x)+a_i\vartheta\langle\nabla \J^{\e}(\hat\x),x\uv_i-\hat\x\rangle+o(\vartheta)+G_i((1-\vartheta)\hat x_i+\vartheta x),\\
    &\ge a_i\J^{\e}(\hat\x)+a_i\vartheta\langle\nabla \J^{\e}(\hat\x),x\uv_i-\hat\x\rangle+o(\vartheta)+(1-\vartheta)G_i(\hat x_i)+\vartheta G_i(x),
\end{align}
where the second equality used the Taylor expansion of $\J^{\e}$, while the second inequality used the concavity of $G$. Simplifying, we obtain,
\begin{align}
    a_i\vartheta\langle\nabla \J^{\e}(\hat\x),x\uv_i-\hat\x\rangle+o(\vartheta)+\vartheta G_i(x)\le \vartheta G_i(\hat x_i).
\end{align}
Dividing by $\vartheta$ and letting $\vartheta\to 0$, we obtain,
\begin{align}
    a_i\langle\nabla \J^{\e}(\hat\x),x\uv_i-\hat\x\rangle+G_i(x)\le G_i(\hat x_i).
\end{align}
Since this is true for any $x\in\mathcal X_i$, we can sum it up over $i$ to obtain,
\begin{align}\label{eq:sum-cond-1}
    \langle\nabla \J^{\e}(\hat\x),\x-\hat\x\rangle+\sum_{i=1}^N\frac{1}{a_i}G_i(x_i)\le \sum_{i=1}^N\frac{1}{a_i}G_i(\hat x_i),~\forall\x\in\mathcal X.
\end{align}
Since $\Phi^{\e}$ is  concave, we have
\begin{align}
    \Phi^{\e}(\x) &=\J^{\e}(\x)+\sum_{i=1}^N \frac{1}{a_i}G_i(x_i),\\
    &\le \J^{\e}(\hat\x)+\langle\nabla \J^{\e}(\hat\x),\x-\hat\x\rangle+\sum_{i=1}^N\frac{1}{a_i}G_i(x_i).
\end{align}
Using \eqref{eq:sum-cond-1} in this, we obtain 
\begin{align}
    \Phi^{\e}(\hat\x)\ge \Phi^{\e}(\x),~\forall~\x\in\mathcal X,
\end{align}
which implies that $ \hat\x\in\Gamma^{\e}$, and hence
\begin{align}
   \mathbf N^{\e}\subseteq\Gamma^{\e}.
\end{align}\done
\end{proof}
Thus, under the separable structure of utilities that we are considering, whenever the collective utility $\J$ is differentiable, the potential function captures all the Nash equilibria of the game in its set of maximizers. The following corollary is obtained by combining this result with Lemma \ref{lem:gamma-properties}.
\begin{cor}
    If $\J$ is differentiable and $\e>0$, $|\mathbf N^{\e}|=1$.
\end{cor}
If $\J$ is not differentiable, the sets $\mathbf N^{\e}$ and $\Gamma^{\e}$ may not be identical (see Example \ref{ex:non-conv-bestresp} in Section \ref{subsec:IBR-conv}).

\section{Characterizing Regularized Equilibria}\label{sec:characterize}
\noindent We relate the regularized potential maximizing equilibria of $\G^{\e}$ to the equilibria of  the unregularized game $\G^{0}$, in terms of the regularizing function $H$. Define the state
\begin{align}
   \x^{\dagger}\triangleq \arg_{\x\in\Gamma^0}\max H(\x). 
\end{align}
Since $\Gamma^0$ is non empty, convex and compact (Lemma \eqref{lem:gamma-properties}), and since $H$ is strongly concave (Assumption \ref{assumption:concave}), existence of such a unique $\x^{\dagger}$ is clear~\cite[Proposition 3.1.1]{bertsekas2009convex}.
\begin{lem}\label{lem:eps-eq-reg}
    Any element $\x^{\e}=(x_i^{\e},\x_{-i}^{\e})$ of~$\Gamma^{\e}$ is an $\epsilon$-Nash equilibrium for $\G^0$, where
    \begin{align}\epsilon=-\e a_1  H(\x^{\dagger}).\end{align}
    Further, the potentials satisfy
    \begin{align}\Phi^0(\x^{\dagger})\ge \Phi^0(\x^{\e})\ge\Phi^0(\x^{\dagger})-\frac{\epsilon}{a_1}.\end{align}
\end{lem}
\begin{proof}
    Since $\Phi^{\e}=\Phi^0+\e H$ with $H$  negative, we have
    \begin{align}\label{eq:phi-0-e-bd1}
        \Phi^{0}(\x^{\e})\ge \Phi^{\e}(\x^{\e}).
    \end{align}
    Since $\x^{\e}\in\Gamma^{\e}$,
    \begin{align}\label{eq:phi-0-e-bd2}
        \Phi^{\e}(\x^{\e})\ge \Phi^{\e}(\x^{\dagger})=\Phi^0(\x^{\dagger})+\e H(\x^{\dagger}).
    \end{align}
    Now we have
    \begin{align*}
        J_i^0(x_i^{\e},\x_{-i}^{\e})-J_i^0(x,\x_{-i}^{\e}) &=a_i(\Phi^0(\x^{\e})-\Phi^0(x,\x_{-i}^{\e})),\\
        &\ge a_i(\Phi^0(\x^{\dagger})-\Phi^0(x,\x_{-i}^{\e})\\ &~~~~+\e H(\x^{\dagger})),\\
        &\ge -\epsilon,
    \end{align*}
    where we used \eqref{eq:phi-0-e-bd1},\eqref{eq:phi-0-e-bd2} and the fact that $\x^{\dagger}\in\Gamma^0$. Thus we see that $\x^{\e}$ is an $\epsilon-$equilibrium for the game $\G^0$. The relationship between the potential  values also follows from the definition of $\x^{\dagger}$, \eqref{eq:phi-0-e-bd1} and \eqref{eq:phi-0-e-bd2}.\done
\end{proof}
By choosing an appropriate $H$, one can choose which equilibrium of $\Gamma^0$ is close to the new regularized equilibrium. Here the closeness is in terms of the potential function $\Phi^0$. For a given $H$, we see an effect similar to the equilibrium selection  problem 
\begin{align}
    \max H(\x)~ s.t.\x\in\Gamma^0,
\end{align}
where $\Gamma^0$ is from \eqref{def:gamma-set}.
Among all elements $\x$ of $\Gamma^0$, $\x^{\dagger}$ minimizes a lower bound on the difference between $\Phi^0(\x)$ and $\Phi^0(\x^{\e})$. The utility form in \eqref{eq:utility-fn-def} therefore can be viewed as a relaxed form of the equilibrium selection problem.

\section{Learning Dynamics for Nash Equilibria}\label{sec:Learning}
\noindent There are multiple approaches to learning equilibria in a noncooperative game, which have been explored in the literature. Two important methods are {gradient ascent/play} and {iterative best response}. 
\begin{defn}
    Gradient ascent/play (GA) is given by the sequence of joint strategies $(\x(k),k\ge 0)$ with $\x(k):=[x_1(k),\cdots,x_N(k)]$ given by
    \begin{align}
    x_i(k)=x_i(k-1)+\Delta_i(k)\frac{\partial J_i^{\e}(x_i,\x_{-i})}{\partial x_i}{|}_{\x=\x_{k-1}},
    \end{align}
with $\Delta_i(k)>0$ being an appropriate step size.
\end{defn}
\begin{defn}\label{def:IBR}
    Iterative best response (IBR) is given by the sequence of joint strategies $(\x(k),k\ge 0)$ where the vector $\x(k):=[x_1(k),  \cdots,x_N(k)]$ is given by 
    \begin{align}
    \x(k+1)=(x_1(k+1),...,x_N(k+1)),
    \end{align}
    where
\begin{align}
x_i({k+1})=\mathscr B_i^{\e}\left(x_1(k+1),...,x_{i-1}(k+1),x_i(k),...,x_{N}(k)\right),
\end{align}
with $\mathscr B_i^{\e}:\mathcal X\to \mathcal X_i$ being the best-response map given by
    \begin{align}
        \mathscr B_i^{\e}(\x)=\arg_{y\in\mathcal X_i}\max J_i^{\e}(y,\x_{-i}).
    \end{align}
\end{defn}
The maximizer exists due to continuity of $J_i^{\e}$ and compactness of $\mathcal X_i$. The IBR dynamics can be interpreted as a coordinate ascent method on the potential function $\Phi^{\e}$, where each agent sequentially updates its strategy to maximize the potential with respect to its own variable.

In addition to these two existing methods, we propose a minorization-maximization (MM) based method for learning equilibria. Here, the algorithm works on the potential function.

\begin{defn}
    Iterative minorize-maximize (IMM) is given by the sequence of joint strategies $(\x(k),k\ge 0)$, where 
    \begin{align}
    \x(k+1)=\arg_{\x\in\mathcal X}\max\phi(\x|\x(k)),
    \end{align}
    where $\phi(\cdot|\x(k))$ is the minorizing function at iteration $k$, which is any function that satisfies,
    \begin{enumerate}
        \item $\phi(\x|\x(k))\le \Phi^{\e}(\x),~\forall\x\in\mathcal X$,
        \item  $\phi(\x(k)|\x(k))= \Phi^{\e}(\x(k))$.
    \end{enumerate}
    The initial point $\x(0)\in\mathcal X$.
\end{defn}
Since the update step of IMM is agnostic to the structure of the potential function, and relies instead on the properties of the minorizing function, it has particular appeal when the utilities of the users are not differentiable.

\section{Convergence of Learning Dynamics}\label{sec:convergence}
In this section, we study the convergence properties of the different learning schemes presented in Section \ref{sec:Learning}, under different assumptions on the structure of the utility functions.
\subsection{Convergence of Gradient Ascent}
\noindent If $\G^{\e}$ is smooth, both gradient play and iterative best response can be shown to be converge to an $\epsilon-$Nash equilibrium of the unregularized game $\G^0$.
\begin{thm}\label{thm:grad-conv}
    Let $\e>0$ and let $\G^{\e}$ be a smooth game. Then, the sequence $(\x(k),k\ge 0)$ generated by gradient play
    converges to the unique Nash equilibrium of $\G^{\e}$. This equilibrium is also an $\epsilon-$Nash equilibrium of $\G^0$.
\end{thm}
\begin{proof}
    Since the utilities are continuously differentiable and strictly concave, the game  is strictly monotone~\cite{rosen1965existence}, and hence there is a unique Nash equilibrium~\cite[Theorem 2]{rosen1965existence} to which gradient dynamics converges~\cite[Theorem 10]{rosen1965existence}. The unique Nash equilibrium is also the unique element of $\Gamma^{\e}$; hence it is an $\epsilon-$Nash equilibrium of $\G^0$, with  $\epsilon$ given by Lemma \ref{lem:eps-eq-reg}.\done
\end{proof}
Note that since for all $i$,
\begin{align}
    a_i\frac{\partial \Phi^{\e}(\x)}{\partial x_i}=\frac{\partial J^{\e}(x_i,\x_{-i})}{\partial x_i},
\end{align}
gradient descent on the utilities and on the potential function are equivalent up to a scaling factor.  For $\e>0$, the gradient method has exponential convergence; however, when the function is non smooth, the rate degrades to sublinear~\cite{bubeck2015convex}.
\subsection{Convergence of Iterative Best Response}\label{subsec:IBR-conv}
To study the IBR sequence, we first define the following auxiliary variables.
\begin{align}
\w^0(k) &=\x(k),\\
    \w^i(k) &=(w^i_1(k),\cdots,w^i_N(k))=(w^i_j(k),\w^i_{-j}(k))\in\mathcal X.
\end{align}
 From the definition of IBR,        
        \begin{align}
        \w^i(k) &=\left(\mathscr B_i^{\e}(\w^{i-1}(k)),\w_{i-1}(k)\right),i=1,...,N,\\
        \x(k+1) &=\w^N(k).
        \end{align}
 We first state a technical lemma which will be used in subsequent proofs.
\begin{lem}\label{lem:str-ccv-prop}
Let $\mathscr F:\mathcal Y\to \mathbb R$ be a strongly concave, differentiable function with modulus $\rho$ with $\mathcal Y\subseteq \mathbb R^M$. Let $\mathscr F$ be maximized at the unique point $\y^*\in\mathcal Y$. Then,
\begin{align}\label{eq:str-ccv-prop}
    ||\y-\y^*||^2\le\frac{2}{\rho}\left(\mathscr F(\y^*)-\mathscr F(\y)\right)~\forall~\y\in\mathcal Y.
\end{align}
\end{lem}
\begin{proof}
    Apply the fact that $\partial\mathscr F(\y^*)$ contains only the element $\nabla\mathscr F(\y^*)=0$, in \eqref{eq:def-str-cvx}.\done
\end{proof}

We now prove two properties of the IBR sequence: first, that it is potential increasing; and two, that adjacent values of the auxiliary variables grow closer to each other. 
\begin{lem}\label{lem:IBR-monotone-potential}
For $\e>0$, the auxiliary variables satisfy the following.
\begin{enumerate}
    \item $
        \Phi^{\e}(\w^i(k))\ge \Phi^{\e}(\w^{i-1}(k)),~i=1,...,N.
    $
    \item $||\w^i(k)-\x(k)||^2\to 0,~i=1,...,N.$
\end{enumerate}
\end{lem}
\begin{proof}
    From the definition of best response, we have
        \begin{align}
        J_i^{\e}(w_i^i(k),\w^i_{-i}(k))\ge J_i^{\e}(w^{i-1}_i(k),\w^{i-1}_{-i}(k)),~i=1,...,N,
        \end{align}
        which implies, by \eqref{eq:pot-equal-util} that
    \begin{align}\label{eq:sandwich-phi}\Phi^{\e}(\w^i(k))\ge \Phi^{\e}(\w^{i-1}(k)),~i=1,...,N.\end{align}
    Due to the fact that $(\Phi^{\e}(\x(k)),k\ge 0)$ is a monotone sequence, it has a limit; let it be $\overline\Phi^{\e}$. From \eqref{eq:sandwich-phi}, it also follows that 
    \begin{align}
    \label{eq:phi-limiting}
    \lim_{k\to\infty}\Phi^{\e}(\w^i(k))=\overline\Phi^{\e},~i=1,...,N.
    \end{align}
    Define
    \begin{align}
         \psi(x)=\Phi^{\e}(x,\w^{i-1}_{-i}(k)).
    \end{align}
      Under Assumption \ref{assumption:concave}, it follows that $\psi(x)$ is strongly concave with modulus $\e\mu$. Note that $\psi(x)$ is maximized at $x=w_i^i(k)$. Using \eqref{eq:str-ccv-prop}, we have
      \begin{align}\label{eqn:psi-ineq1}
        \Phi^{\e}(\w^{i-1}(k))-\Phi^{\e}(\w^{i}(k)) &= \psi(w^{i-1}_i(k))-\psi(w^{i}_i(k))\\
        &\le -\frac{\e\mu}{2}(w^{i-1}_i(k)-w^{i}_i(k))^2\\
        &=-\frac{\e\mu}{2}(x_i(k)-x_i(k+1))^2.\label{eqn:psi-ineq3}
      \end{align}
    Hence it follows that
    \begin{align}
        ||\w^i(k)-\x(k)||^2 &=\sum_{j\le i}(x_j(k+1)-x_j(k))^2\\
        &\le \frac{2}{\e\mu}\sum_{j\le i}(\Phi^{\e}(\w^i(k))-\Phi^{\e}(\w^{i-1}(k))),\label{eqn:phi-bd}
    \end{align}
      where \eqref{eqn:phi-bd} followed from \eqref{eqn:psi-ineq1}-\eqref{eqn:psi-ineq3}. Applying \eqref{eq:phi-limiting} to \eqref{eqn:phi-bd}, we see that 
      \begin{align}\label{eq:w-x-con}
      ||\w^i(k)-\x(k)||^2\to 0,~i=1,...,N.
      \end{align}\done
\end{proof}
The set of limit points generated by IBR has the following properties.
\begin{lem}\label{lem:IBR-subseq-compact}
For $\e>0$, let $\Omega$ be the set of limit points generated by IBR, i.e.,
    \begin{align}
        \Omega=\{\x\in\mathcal X|\exists k_j\to\infty~s.t.~\x(k_j)\to\x\},
    \end{align}
    for $(\x(k),k\ge 0)$ generated as per definition \ref{def:IBR}.
    Then, $\Omega$ is non empty, compact and connected.
\end{lem}
\begin{proof}
   From Lemma \ref{lem:IBR-monotone-potential}, it follows that
    \begin{align}
       \Phi^{\e}(\x(k+1))\ge \Phi^{\e}(\x(k))~\forall~k.
    \end{align}
    Therefore, iterates $(\x(k))_{k\ge 0}$ belong to the set
\begin{align}
\mathcal L=\{\x\in\mathcal X| \Phi^{\e}(\x)\ge\Phi^{\e}(\x(0))\}.
\end{align}
Since $\Phi^{\e}(\x)$ is a concave function, and since the set of maximizers of $\Phi^{\e}(\x)$ is a singleton (by Lemma \ref{lem:gamma-properties}), we see, using \cite[Proposition B.9]{bertsekas1999nonlinear} that $\mathcal L$ is a compact set. Hence there is a subsequence $(\x(k_j))$ with a limit point $\hat\x\in\mathcal X$. Hence $\Omega$ is non empty.
From Lemma \ref{lem:IBR-monotone-potential} and \cite[Lemma 2.6]{bauschke2015fej}, it follows that $\Omega$ is compact and connected.\done
\end{proof}
Now we show that every limit point of the IBR sequence is a Nash equilibrium.
\begin{thm}\label{thm:IBR-prop-jay}
    Let $\hat\x\in\Omega$ for $\e>0$. Then $\hat\x$ is a Nash equilibrium of the game $\G^{\e}$.
\end{thm}
\begin{proof}
    Let $(\x(k_j))$ be a convergent subsequence of the IBR sequence, with limit $\hat\x$. Since
    \begin{align}\w^1(k_j)=(\mathscr B_1(\x(k_j)),\x_{-1}(k_j)),
    \end{align}
    we have
    \begin{align}
    J_1^{\e}(\w^1(k_j))\ge J_1^{\e}(x,x_2(k_j),...,x_N(k_j))~\forall x\in\mathcal X_1,
    \end{align}
    taking the limit along the subsequence $k_j\to\infty$ and using \eqref{eq:w-x-con}, we have
    \begin{align}\label{eqn:J-bound-pf}
    J_1^{\e}(\hat\x)\ge J_1^{\e}(x,\hat x_2,...,\hat x_N)~\forall x\in\mathcal X_1.
    \end{align}
    Repeating the argument for $\w^i(k_j)$ for other values of $i$, we see that $\hat\x$ satisfies, for all $i=1,...,N$,
    \begin{align}
        J_i^{\e}(\hat\x)\ge J_i^{\e}(x,\hat\x_{-i})~\forall~x\in\mathcal X_i.
    \end{align}
    Thus $\hat\x\in\mathbf N^{\e}$.
     \done
\end{proof}
While all limit points are Nash equilibria, they need not be in the set $\Gamma^{\e}$. An example is presented below.
\begin{exmp}\label{ex:non-conv-bestresp}
    Consider a two agent team game\footnote{A team game (or team problem) is a game in which all agents have identical utility~\cite{basar1998dynamic}.} with utilities
    \begin{align}
    J_1^{\e}(x_1,x_2)=J_2^{\e}(x_1,x_2)=5-|x_1+x_2|-\e (x_1^2+x_2^2),
    \end{align}
    with $\e=0.1$ and $\mathcal X_1=\mathcal X_2=[-100,100]$. The utilities are strictly concave and non differentiable, and the game has potential function
    \begin{align}
    \Phi^{\e}(x_1,x_2)=5-|x_1+x_2|-\e (x_1^2+x_2^2).
    \end{align}
    Since $\Phi^{\e}(x_1,x_2)<5$ for all $(x_1,x_2)\ne (0,0)$ and $\Phi^{\e}(0,0)=5$, we have
    \begin{align}
    \Gamma^{\e}=\{(0,0)\}.
    \end{align}
    However, the game has  other Nash equilibria. For example, $\x=(4,-4)$ is a Nash equilibrium of the game, as are many points on the line $x_1+x_2=0$. A best-response iteration that starts with agent 1 having its strategy being 4, leads to the limit point Nash equilibrium $(4,-4)$, which is not in $\Gamma^{\e}$.
\end{exmp}

It can be shown that if the collective benefit term $\J$ is differentiable, iterated best response converges to the potential optimal equilibrium. Note that this includes the case where the entire utility function $J_i^{\e}$ is smooth.

\begin{thm}\label{thm:separable-non-diff-case}
  Let $\J$ be differentiable. Let $(\x(k),k\ge 0)$ be iteratively generated by the best-response iteration as described above, in the game $\G^{\e}$ for some $\e>0$. Then,  $\x(k)$ converges to the unique Nash equilibrium of $\G^{\e}$.
\end{thm}
\begin{proof}
   From Lemma \ref{lem:IBR-subseq-compact}  we have that there is a convergent subsequence $(\x(k_j),j\ge 0)$ of $(\x(k),k\ge 0)$. Let $\hat\x$ be the limit of this subsequence.  From Theorem \ref{thm:IBR-prop-jay}, $\hat\x$ is a Nash equilibrium. Since $\J$ is differentiable, using  Lemma \ref{lem:separable-nash-potential-connex}, it follows that $\hat\x$ is a potential optimal equilibrium, i.e.,
   \begin{align}
       \Phi_i^{\e}(\hat\x)\ge \Phi_i^{\e}(\x),~\forall~\x\in\mathcal X,
   \end{align}
   and thus, 
   \begin{align}
        \Phi_i^{\e}(\hat\x)=\Phi^{\e\star}.
   \end{align}
 From Lemma \ref{lem:IBR-monotone-potential}, it follows that the sequence $(\Phi^{\e}(\x(k)),k\ge 0)$ is monotone, and hence, convergent. Since it converges to $\Phi^{\e\star}$ along the subsequence $(\x(k_j),j\ge 0)$, it follows that
\begin{align}\label{eq:phi-lam-star}
    \Phi^{\e}(\x(k))\to\Phi^{\e\star},
\end{align}
as $k\to\infty$. Using the inequality \eqref{eq:str-ccv-prop} along with the limit  \eqref{eq:phi-lam-star}, we have
\begin{align}
    ||\x(k)-\x^{\e}||^2\le \frac{2}{\mu}\left(\Phi^{\e\star}-\Phi^{\e}(\x(k))\right)\to 0~\text{as~}k\to\infty.
\end{align}
Thus $(\x(k),k\ge 0)$ converges to the unique Nash equilibrium $\x^{\e}$ of $\G^{\e}$, which is also an $\epsilon-$ Nash equilibrium of $\G^0$ (using Lemma \ref{lem:eps-eq-reg}).\done
\end{proof}

Hence, the differentiability of the collective benefit term is crucial in allowing the convergence of IBR.

\subsection{Convergence of Iterative MM}
\noindent From the properties of IMM, it follows that with every step of the MM iteration, we have
\begin{align}\label{eq:IMM-incr-prop}
    \Phi^{\e}(\x(k)) &= \phi(\x(k)|\x(k))\le \phi(\x(k+1)|\x(k))\le \Phi^{\e}(\x(k+1)).
\end{align}
Thus the potential is non decreasing over the MM iterates. 

In the rest of this paper, we consider the following implementation of IMM. The sequence of minorizers is defined by
\begin{align}\label{eq:IMM-def}
    \phi(\x|\x(k))=\Phi^{\e}(\x)-\beta||\x-\x(k)||^2.
    \end{align}
The MM dynamics can be interpreted as a majorization-minimization procedure applied to the potential function. At each iteration, the algorithm performs an exact maximization of a surrogate function that globally lower-bounds the potential while being tight at the current iterate. As a consequence, the sequence of potential values is non-decreasing and converges to the global maximum of $\Phi^{\e}$ over the compact strategy set $\mathcal X$, implying that all limit points lie in $\Gamma^{\e}$.

Note that these surrogates are agnostic to the structure of the potential function being optimized. It is also an example of a Bregman minorizer~\cite{lange2021nonconvex}. Noting that the IMM algorithm is also a proximal algorithm~\cite{parikh2014proximal}, we obtain the following convergence result.

\begin{prop}
     Let $(\x(k),k\ge 0)$ be a sequence of iterates generated by IMM with minorizers chosen by \eqref{eq:IMM-def}. Then, 
     \begin{enumerate}
         \item (For $\e>0$): The iterates  converge exponentially  to $\x^{\e}$, the unique potential optimal Nash equilibrium of $\G^{\e}$, contained in  $\Gamma^{\e}$, i.e., there exist constants $C>0$, $\gamma\in(0,1)$ such that for all $k$,
         \begin{align}
         ||\x(k)-\x^{\e}||\le C\gamma^{k}.
         \end{align}
         \item (For $\e=0$): The iterates $(\x(k),k\ge 0)$ converge to an element of $\Gamma^{\e}$. Further, each $\x(k)$ is an $\epsilon(k)-$ equilibrium of $\G^0$, with $\epsilon(k)\to 0$ as $k\to\infty$.
     \end{enumerate}
\end{prop}

\begin{proof}
(For $\e>0$):
From assumption \ref{assumption:concave}, it follows that $\Phi^{\e}(\x)$ is strongly concave with modulus $\e\mu$, and hence for all $\x_1,\x_2\in\mathcal X$,
\begin{align}\label{eqn:strongConcaveIdentity}
\Phi^{\e}(\x_1)-\Phi^{\e}(\x_2)\le \langle\x_1-\x_2,\y\rangle-\frac{\e\mu}{2}||\x_1-\x_2||^2,~\forall \y\in\partial\Phi^{\e} (\x_2),\end{align}
where $\partial\Phi^{\e}(\x_2)$ is the set of subgradients  of the function $\Phi^{\e}(\cdot)$ at $\x_2$. Let $k$ be arbitrary. 
Since we have 
\begin{align}
\x(k+1)=\arg_{\x}\max\phi(\x|\x(k)),
\end{align}
it follows that $\textbf{0}\in\partial\phi(\x(k+1)|\x(k))$. By the sum rule for the subdifferential~\cite[Theorem 23.8]{rockafellar1970convex},
\begin{align}
\partial\phi(\x(k+1)|\x(k))=\partial\Phi^{\e}(\x(k+1))-2\beta(\x(k+1)-\x(k)).
\end{align}
Therefore, there exists $\nu\in\partial\Phi^{\e}(\x(k+1))$ such that
\begin{align}\label{eqn:nu-property}
 \nu=2\beta(\x(k+1)-\x(k).   
\end{align}
Using \eqref{eqn:strongConcaveIdentity} with $\x_1=\x^{\e}$, $\x_2=\x(k+1)$ and $\y=\nu\in\partial \Phi^{\e}(\x(k+1)))$, we obtain the inequality
\begin{align}
   \Phi^{\e}(\x^{\e})-\Phi^{\e}(\x(k+1)))\le  \langle\x^{\e}-\x(k+1),\nu\rangle-\frac{\e\mu}{2}||\x^{\e}-\x(k+1)||^2.
\end{align}
Now, using \eqref{eqn:strongConcaveIdentity} with $\x_1=\x(k+1)$ and $\x_2=\x^{\e}$, while noting that $\textbf{0}\in\partial  \Phi^{\e}(\x^{\e}) $, we obtain the inequality
\begin{align}
   \Phi^{\e}(\x(k+1))-\Phi^{\e}(\x^{\e}))\le  \langle\x(k+1)-\x^{\e},\textbf{0}\rangle-\frac{\e\mu}{2}||\x^{\e}-\x(k+1)||^2.
\end{align}
Adding the two equations we get 
\begin{align}
\e\mu||\x^{\e}-\x(k+1)||^2\le \langle\x^{\e}-\x(k+1),\nu\rangle &=2\beta\langle\x^{\e}-\x(k+1),\x(k+1)-\x(k)\rangle,\\
&=2\beta\left(\langle\x^{\e}-\x(k+1),\x^{\e}-\x(k)\rangle-||\x^{\e}-\x(k+1)||^2\right), 
\end{align}
where the first equality used \eqref{eqn:nu-property}. 
Simplifying, we obtain,
\begin{align}\left(1+\frac{\e\mu}{2\beta}\right)||\x^{\e}-\x(k+1)||^2 &\le \langle\x^{\e}-\x(k+1),\x^{\e}-\x(k)\rangle\\ &\le ||\x^{\e}-\x(k+1)||\cdot||\x^{\e}-\x(k)||,\end{align}
where the second inequality follows from the  Cauchy-Schwarz inequality. Thus, we obtain
\begin{align}
||\x^{\e}-\x(k+1)||\le  \left({1+\frac{\e\mu}{2\beta}}\right)^{-1}||\x^{\e}-\x(k)||. 
\end{align}
Using this equation iteratively, we obtain the result with $\gamma=\left({1+\frac{\e\mu}{2\beta}}\right)^{-1}$ and $C=||\x^{\e}-\x(0)||$.

\noindent (For $\e=0$): With $\e=0$, the IMM iterates can be viewed as the proximal maximization of a concave function. Let $\x^*\in\Gamma^0$. Since $\Phi^0$ is a concave function being maximized at $\x^*$ with $\Phi^0(\x^*)=\Phi^{0\star}$, we have
\begin{align}
    \Phi^{0\star}-\Phi^0(\x(k+1))\le \langle \nu, \x^*-\x(k+1)\rangle~\forall \nu\in\partial \Phi^0(\x(k+1)).
\end{align}
From the definition of $\x(k+1)$, it follows that it follows that $\textbf{0}\in\partial\phi(\x(k+1)|\x(k))$. Using the sum rule for the subdifferential~\cite[Theorem 23.8]{rockafellar1970convex} and simplifying, it follows that $\beta(\x(k+1)-\x(k))\in\partial \Phi^0(\x(k+1) $. Hence it follows that,
\begin{align}
    \Phi^{0\star}-\Phi^0(\x(k+1))\le \beta\langle \x(k+1)-\x(k), \x^*-\x(k+1)\rangle.
\end{align}
Using the vector identity $2\langle\y(1)-\y(2),\y(3)-\y(1)\rangle=||\y(3)-\y(2)||^2-||\y(1)-\y(2)||^2-||\y(3)-\y(1)||^2$, the inequality becomes,
\begin{align}
    \Phi^{0\star}-\Phi^0(\x(k+1)) &\le \frac{\beta}{2}\left(||\x^*-\x(k)||^2-||\x^*-\x(k+1)||^2-||\x(k+1)-\x(k)||^2\right),\\
    &\le \frac{\beta}{2}\left(||\x^*-\x(k)||^2-||\x^*-\x(k+1)||^2\right).\label{eq:phi-x-norm-identity}
\end{align}
By adding up and using the telescoping property, we get
\begin{align}\label{eq:phi-key-inew}
    \sum_{j=0}^{k-1}\left(\Phi^{0\star}-\Phi^0(\x(j+1))\right)\le \frac{\beta}{2}||\x^*-\x(0)||^2.
\end{align}
From \eqref{eq:IMM-incr-prop}, we know that
\begin{align}
    \sum_{j=0}^{k-1}\left(\Phi^{0\star}-\Phi^0(\x(j+1))\right)\ge k\left(\Phi^0(\x^*)-\Phi^0(\x(k)\right).
\end{align}
Using this in \eqref{eq:phi-key-inew}, we have that
    \begin{align}\label{eq:phi-k-bound-mm}
    \Phi^0(\x(k))\ge \Phi^{0\star}-\frac{\beta ||\x(0)-\x^*||^2}{2k}.
    \end{align}
Thus, we have that
\begin{align}
\lim_{k\to\infty}\Phi^0(\x(k))=\Phi^{0\star}.
\end{align}
Since the sequence $(\x(k),k\ge 0)$ takes values in the compact set $\mathcal X$, there is a convergent subsequence $(\x(k_j),j\ge 0)$ with limit point $\hat\x\in\mathcal X$. Since $\Phi^0(\x(k))$ is converging to $\Phi^{0\star}$, it follows that
\begin{align}
    \Phi^0(\hat\x)=\Phi^{0\star},
\end{align}
and consequently, $\hat\x\in\Gamma^0$. From \eqref{eq:phi-x-norm-identity} it follows that $||\x^*-\x(k)||$ is decreasing in $k$ for any $\x^*\in\Gamma^0$; in particular it holds for $\x^*=\hat\x$. Therefore, the sequence $(||\hat\x-\x(k)||,k\ge 0)$ is convergent. However the subsequence $(||\hat\x-\x(k_j)||,j\ge 0)$ this converges to zero. Hence it follows that
\begin{align}
    \lim_{k\to\infty}||\hat\x-\x(k)||=0.
\end{align}
Using \eqref{eq:phi-k-bound-mm} and proceeding as in the proof of Lemma \ref{lem:eps-eq-reg}, we also have that $\x(k)$ is an $\epsilon(k)-$equilibrium of $\G^0$, where
    \begin{align}
    \epsilon(k)=\frac{\beta ||\x(0)-\x^*||^2}{2k}.
    \end{align}
    \done
\end{proof}
Thus, the use of a regularizer speeds up the convergence. The tradeoff is that it converges to an $\epsilon-$equilibrium of the original game, with $\epsilon$ from Lemma \ref{lem:eps-eq-reg}. If there is no regularizer, there is convergence an equilibrium of the original, unregularized game $\G^0$.

We also note that convergence can also be shown with other forms of the MM surrogate if $\Phi^0$ is smooth and strictly concave~\cite{rockafellar1976monotone}, or if the surrogates are differentiable and the algorithm map satisfies some technical conditions~\cite{lange2021nonconvex}.

\section{Numerical Examples}\label{sec:Numerics}
\noindent In this section, we compare the different algorithms that have been introduced in Section \ref{sec:Learning}, by means of different numerical examples. We consider smooth and non smooth games separately. For both of these classes, we provide examples that compare the convergence behavior of different leaning algorithms. We also provide models for smooth and non smooth games arising from practical problems, and the behavior of different learning algorithms applied to these problems.

\subsection{Smooth Games}

\subsubsection*{Convergence Behavior of Learning Algorithms}
\noindent Consider a smooth game with $N=2$ agents. The utilities are of the form
\begin{align}
J_i^{\e}(x_i,x_{-i})=10-(x_1+x_2-2)^2+\mathbb{E}[V_i(\mathcal R_i(x_i))]-\e(x_1^2+x_2^2),
\end{align}
with $\e=0.1$, $\mathcal R_i(x_i)=x_i\xi-d_i$ where $\xi$ is supported over the set $\{2,10\}$ with distribution $(0.8,0.2)$, and
\begin{align}
    V_1(x) =\log(1+x)\mathbf{1}_{(x\ge 0)}+x\mathbf{1}_{(x<0)},~V_2(x) =x.
\end{align}
\begin{figure}
        \centering
        \includegraphics[scale=0.8]{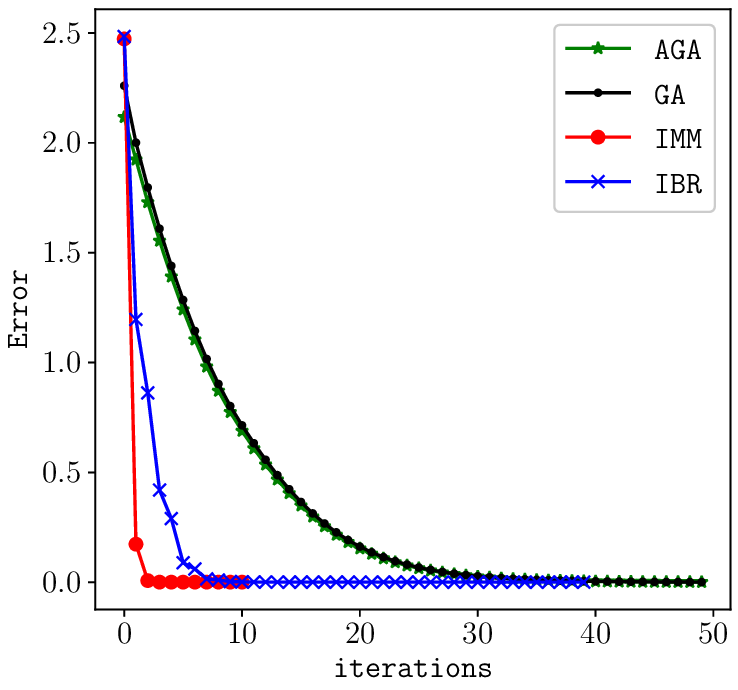}
        \caption{Smooth game.}
        \label{fig:smooth-nums}
\end{figure}
The game has a unique Nash equilibrium, which also coincides with the optimal point of the potential function. We plot the error (i.e., distance to Nash equilibrium) comparing gradient descent (GA), iterative best response (IBR) and iterative MM (IMM) in Fig. \ref{fig:smooth-nums}, starting from the same initial point. We also compare these with the accelerated gradient ascent (AGA) method. From the plot it is clear that IMM converges faster than both IBR and GA to the Nash equilibrium, and in fewer number of iterations. The accelerated gradient method does not offer a huge improvement over the traditional gradient ascent, in this example. While IBR has lower inter-agent communication complexity that GA, calculating the best response at each agent,  may require a higher computational complexity at each agent, as opposed to gradient computation. Depending on the structure of the utility function, this may have differing complexity. In some cases it will be equivalent to solving a linear or quadratic equation.  In other instances it may be that the agents have to use gradient or other learning methods locally in order to compute their best response. The minorizers for each stage of IMM were created as in  \eqref{eq:IMM-def}, with $\beta=0.1$. We note that computing the number of steps needed for each stage of the IMM will also depend on the structure of the optimization at hand. For some cases, such as separable $\J$, each IMM stage involves solving a single equation. In this case the number of iterations is equivalent to the actual running time. In other cases, it may involve solving a fixed point equation or even running a gradient or second order ascent.

\subsubsection*{Steering Collective Outcomes through Incentive Design}

\noindent  This example is adapted from \cite{ashok2025achieving}, which uses a gradient based approach to solve the problem. Consider an example of a collective of users 
buying electricity from an electricity producer. Each user buys $x_i$ units from the producer at a price $p_i$ per unit. Each user has a consumption threshold $d_i$. This could, for example, represent a minimum consumption level to meet their requirements. If the purchased quantity is above this threshold, then the user benefits from the purchase; otherwise it is a loss. The individual random reward is 
\begin{align}
    \mathcal R_i(x_i,\xi)=(x_i-d_i)\xi,
\end{align}
where $\xi$ models randomness of outcome. The collective usage benefit is
\begin{align}
    \mathcal J(\x)=-\sum_{i=1}^N(x_i-d)^2,~\mathcal R_i(x_i,\xi)=(x_i-d_i)\xi.
\end{align}
This represents an attempt by the collective of the users to keep the purchase vector $\x$ close to some desired vector $(d,...,d)$. Here $d$ represents a target energy consumption, for example, mapped from target carbon emission levels. The goal is to design incentives that drive the collective usage benefit $\J(\x)$ to a desired target $\tau$. We note that it in this case, $\J(\x)$ is in itself strictly convex. We define the incentivized utility of the system to be 
\begin{align}
J_i(x_i,\x_{-i})=\J^{\e_i}(\x)+\mathbb{E}[V_i(\mathcal R_i(x_i,\xi))],
\end{align}
where $\J^{\e_i}(\x)=\J(\x)-\e_ix_i$, with $\e_ix_i$ being a linear incentive function (since in this case there is no separate need for a regularizer, since we start with a unique equilibrium). This form of utility is a variant of \eqref{eq:utility-fn-def}, and all  results discussed previously will hold in this case as well. The $\e_i$ are interpreted as per unit incentives/prices to users.

Using iterative MM, one can design a simple learning process to reach the target $\tau$ for $\J(\x)$ as follows. We let the incentives evolve such that they follow a gradient of the function $(\J(\x)-\tau)^2$, while simultaneously proceeding with MM optimization of the potential function of the game. We plot the results for a two player game, with 
\begin{align}
V_1(x)=\log(1+x)\mathbf{1}_{(x\ge 0)}+x\mathbf{1}_{(x<0)},~V_2(x)=x,
\end{align}
and parameters $d=4$, $d_1=1$, $d_2=2$, and $\xi$ supported over the set $\{5,1\}$ with distribution $\{0.8,0.2\}$. We plot three sample paths, all starting from the point $(x_1,x_2)=(1,1)$ with targets $\tau=-4,-4.5$ and $-5$. Using IMM also gets rid of the oscillatory behavior exhibited in the similar example in~\cite{ashok2025achieving}.
\begin{figure}
        \centering
        \includegraphics[scale=0.7]{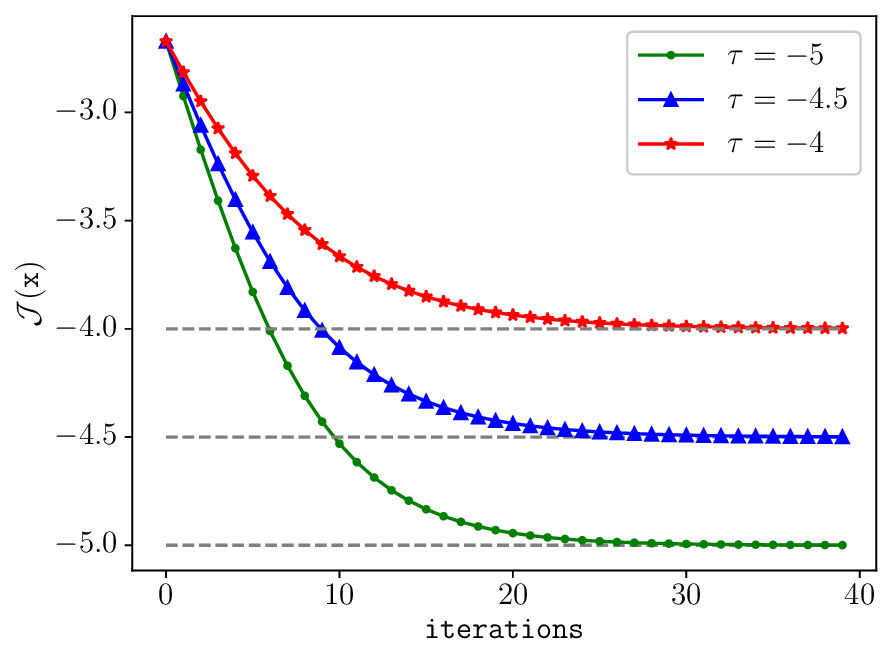}
        \caption{Steering the collective usage function $\J(\x)$ to a desired $\tau$ using IMM.}
        \label{fig:iter-MM-ns}
\end{figure}

\subsection{Non Smooth Games}
\subsubsection*{When Best Responses Fall Short}
\noindent In non smooth games, MM approaches have a distinct advantage over best-response based approaches. 
In order to demonstrate this, 
we consider the same non smooth game that was considered in Example \ref{ex:non-conv-bestresp}, where we had a two agent team game with utilities
    \begin{align}
    J_1^{\e}(x_1,x_2)=J_1^{\e}(x_1,x_2)=5-|x_1+x_2|-\e (x_1^2+x_2^2),
    \end{align}
    with $\e=0.1$. Recall that the potential function corresponding to this game had a unique maximum at $(0,0)$, but the game had other Nash equilibria as well.

    We compare the convergence behavior of iterative MM (IMM),  accelerated (sub) gradient ascent (sGA) and iterative best response (IBR) for this game. We see how the system state $\x=(x_1,x_2)$ evolves starting from the same five initial states, $(2,4)$, $(-4,-4)$, $(-5,4)$, $(10,5)$ and $(10,-1)$. In Fig. \ref{fig:iter-MM-ns}, we plot the iterates of IMM, with the contours of the potential function in the background. The algorithm converges to $(0,0)$, in a rather direct matter, in a few iterations. The algorithm takes steps of large magnitude intitally, and then proceeds in smaller steps.

\begin{figure}
    \centering
    \begin{subfigure}[t]{0.5\textwidth}
        \centering
        \includegraphics[scale=0.5]{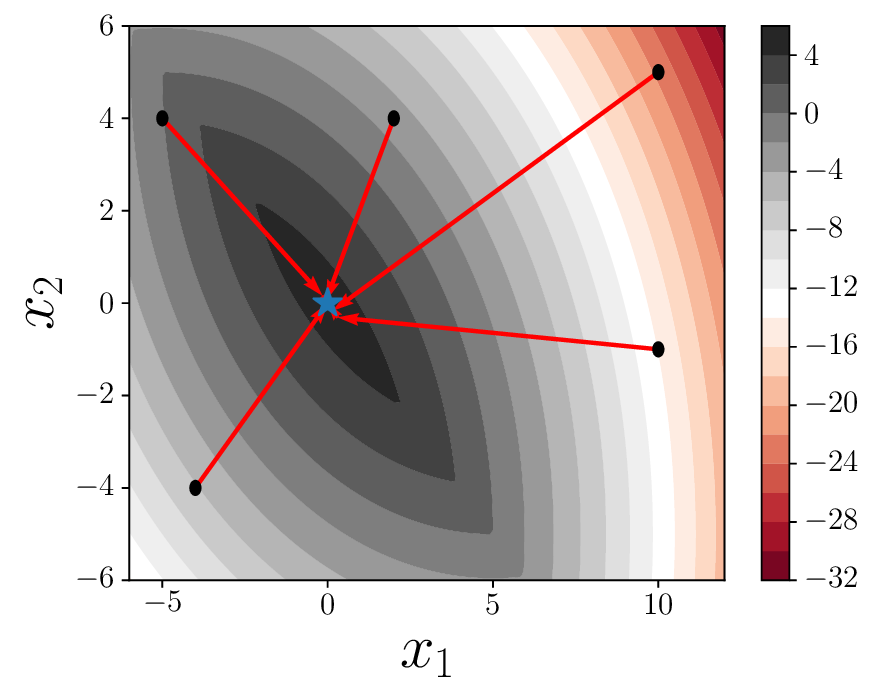}
        \caption{Zoomed out}
    \end{subfigure}%
    ~ 
    \begin{subfigure}[t]{0.5\textwidth}
        \centering
        \includegraphics[scale=0.5]{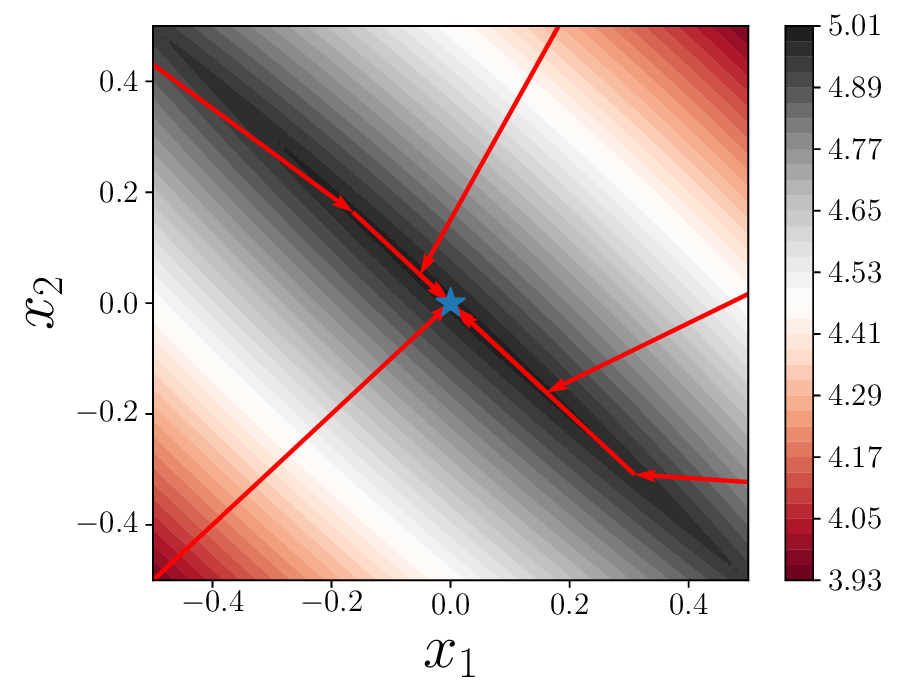}
        \caption{Zoomed in}
    \end{subfigure}
    \caption{State evolution of iterative MM, along the contours of the potential function, starting from different initial states. On the left we have a zoomed out perspective, where we see the larger intial steps of the algorithm. On the right, we zoom in to see the finer steps of the state evolution close to the potential-optimal Nash equilibrium, represented by the blue star at $(0,0)$.}
\label{fig:iter-MM-ns}
\end{figure}
    In the case of sGA, in Fig.  \ref{fig:inter-GA-ns}, we see that they converge to  $(0,0)$, by first reaching the line $x_1+x_2=0$, and then ascending to $(0,0)$.  Note that sGA is the same as GA discussed previously, with the additional property that at non differentiable points, it chooses a value from the set of subdgradients as the derivative, uniformly randomly. We use a fixed small step size of 0.1 for the gradient ascent, because it allows for faster convergence than a decreasing step size. The trajectory shown is averaged over 100 sample paths.
          \begin{figure}
        \centering
        \includegraphics[scale=0.6]{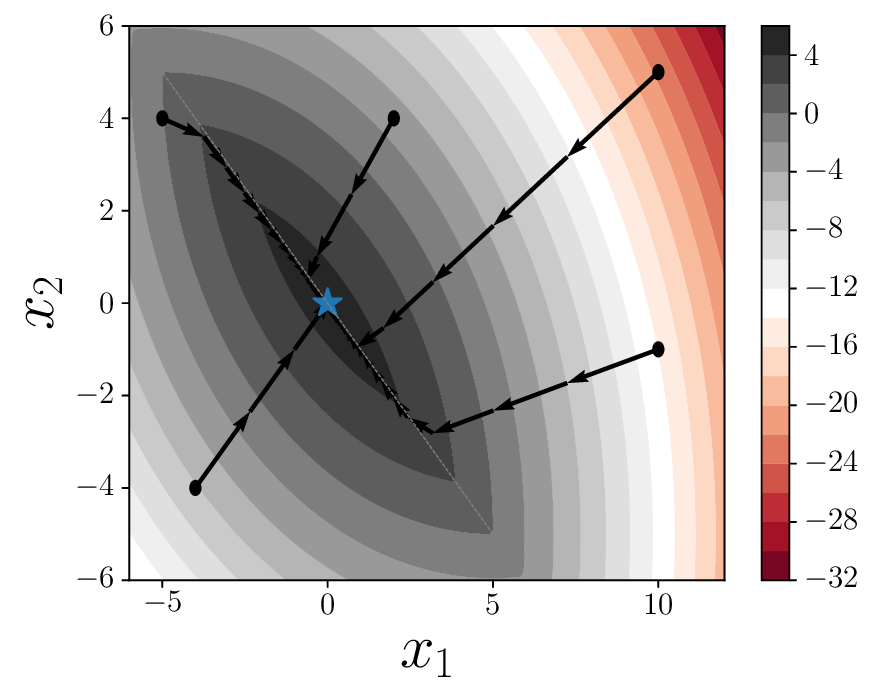}
        \caption{State evolution of sGA, along the contours of the potential function, starting from different initial states.}
        \label{fig:inter-GA-ns}
    \end{figure}
    
    In the case of IBR in Fig. \ref{fig:iter-BR-ns}, we see that iterative best response leads strategies to fixed points that are Nash equilibria of the game, but not optimizers of the potential function. We see that the agents move sequentially, with agent 1 updating their startegy followed by agent 2. They cluster along the line $x_1+x_2=0$. While IBR converges faster than sGA, it cannot move once it reaches the line $x_1+x_2=0$.
\begin{figure}
    \centering
    \includegraphics[scale=0.6]{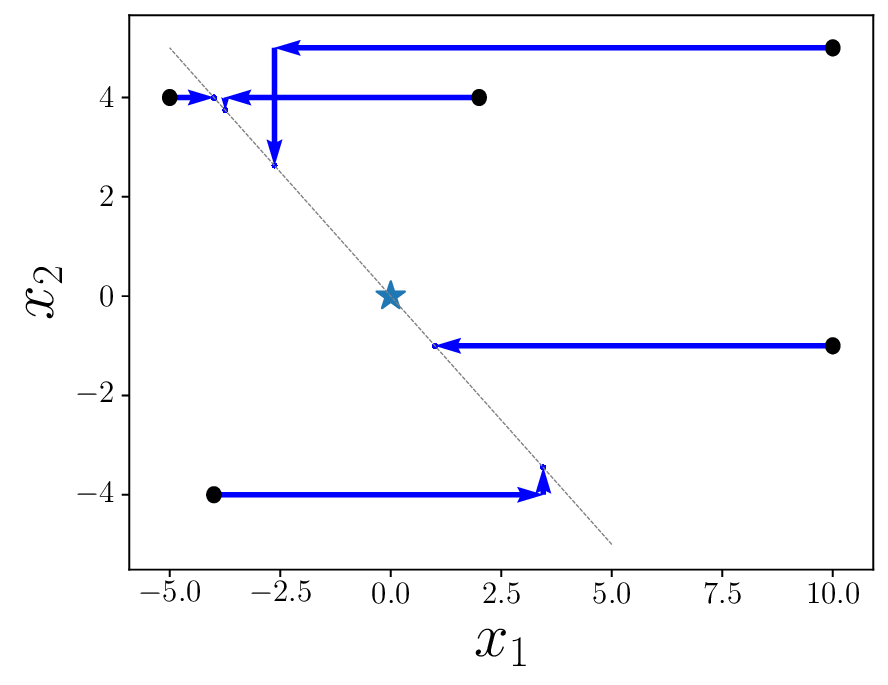}
        \caption{State evolution of iterative BR, converging to points on the line $x_1+x_2=0$, starting from different initial states.}
        \label{fig:iter-BR-ns}
    \end{figure}

   \subsubsection*{Learning to Meet: Distributed Coordination on a Grid}
\noindent We consider the problem of steering multiple agents on a grid to a common point. These could be robotic agents as in the rendezvous problem~\cite{cortes2017coordinated,lin2007multi}, or vehicles coordinating their movement in a city with a Manhattan type layout~
\cite{zhang2023routing}. We consider an example with three agents moving on the positive quadrant. Let the positions (states) of the agents be $\w_i=(x_i,y_i)$ for $i=1,2,3$, and the joint state be $\w=(\w_1,\w_2,\w_3)$. We assume that agents have a cost structure of the form \eqref{eq:utility-fn-def}. Let the collective benefit function be the total absolute error
\begin{align}
   \J(\w)=-\sum_{i=1}^3\sum_{j=1}^{i-1}|x_i-x_j|+|y_i-y_j|, 
\end{align}
which is a measure of closeness of the agents. This benefit is maximized when all the agents coincide in their location. The $L_1$ metric indicates that the agents are constrained to move in a grid. Each agent will also have a cost associated with its current position, which is $C_i(x_i,y_i)$ which captures its perception of its current state. We assume that
\begin{align}
    C_1(x_1,y_1) &=c_1(x_1+y_1),\\
    C_2(x_2,y_2) &=c_2(2-\mathbb{E}\exp(-k_1x_2\xi)-\mathbb{E}\exp(k_2y_2\xi)),\\
    C_3(x_3,y_3) &=c_3(2-\exp(-k_3x_3)-\exp(-k_4y_4)).
\end{align}
Thus agents 2 and 3 have a skewed perception of their cost.
We have the utilities
\begin{align}
J_i(\w_i,\w_{-i})=\J(\w)-\e H(\w)-C_i(x_i,y_i).
\end{align}
Here $H$ is the regularizer, which we choose to be $H(\w)=\sum_{j=1}^3||\w_i||^2$. We have $\e=1$,$c_1=4$, $c_2=0.4$, $c_3=1$, $k_1=0.1$, $k_2=0.3$, $k_3=k_4=1$, and $\xi$ is supported over the set $\{1,100\}$ with distribution $\{0.9,0.1\}$. We consider two approaches to solve this, a distributed version of iterative MM (IMMd) and  (sub) Gradient Ascent (SGA), with initial states $\w_1=(10,0)$, $\w_2=(0,10)$ amd $\w_3=(10,10)$. In Fig. \ref{fig:rob-MM} we show how the three robotic agents move to converge at a common point, using IMMd. In IMMd, the agents proceed in a cyclic fashion. Each one maximizes the MM surrogate locally, for a given state of the other agents. It then passes the updates state information to the next agent, which proceeds likewise.  A coordinating agent is necessary to enable the intermediate communication of system states. The same process using sGA, is shown in Fig \ref{fig:rob-GA}. Using IMM, two of the agents meet much earlier on than in sGA. This suggests that IMM allows for earlier formation of clusters. We see that in sGA, at the termination, the agents do not meet at one point, but continue to move in a loop. This is a result of the synchronized and restricted grid movement with a single step at each instant, and can be remedied by de-synchronizing the agents. 
\begin{figure}
    \centering
    \begin{subfigure}[t]{0.5\textwidth}
        \centering
        \includegraphics[scale=0.5]{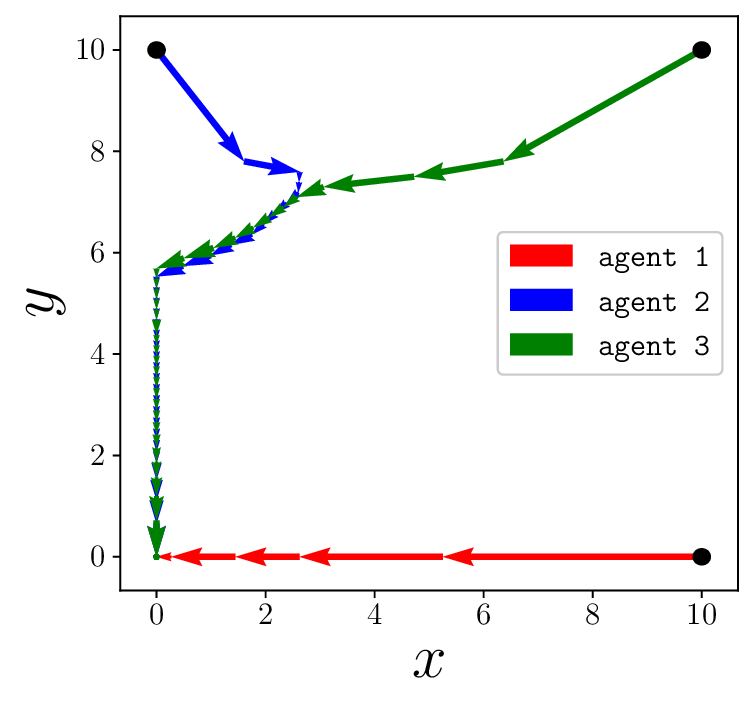}
        \caption{Agents using IMMd}
        \label{fig:rob-MM}
    \end{subfigure}%
    ~ 
    \begin{subfigure}[t]{0.5\textwidth}
        \centering
        \includegraphics[scale=0.5]{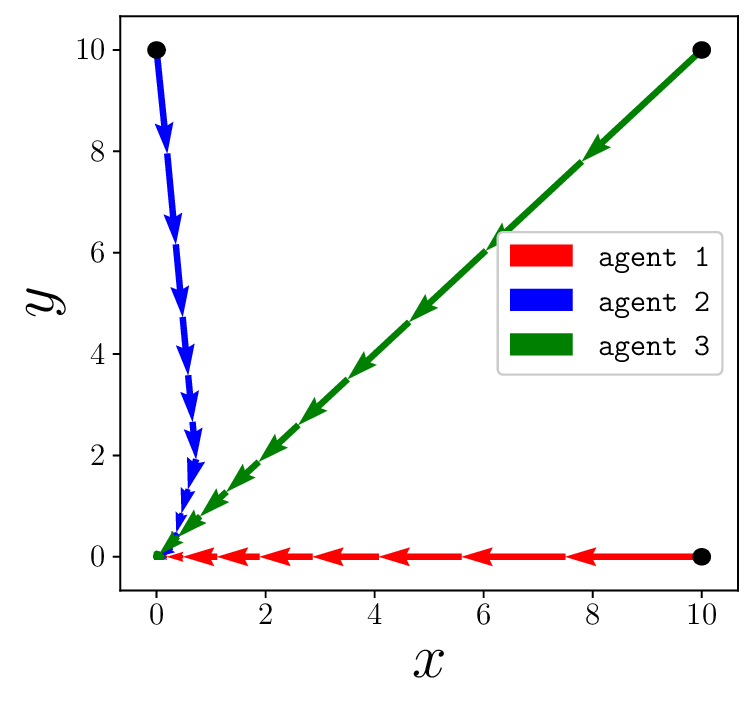}
        \caption{Agents using sGA}
        \label{fig:rob-GA}
    \end{subfigure}
    \caption{Movement of three robotic agents starting from three different locations, for different algorithms.}
\label{fig:rob-move}
\end{figure}

\section{Conclusions and Future Work}
\noindent In this paper, we studied a class of games in which agents derive utility 
from both a shared collective benefit and individually perceived stochastic rewards 
modeled via prospect theory. Such games may admit multiple equilibria, especially 
in the absence of regularization. To address this, we introduced a regularization 
framework that induces strict concavity of the potential function, ensuring the 
existence of a unique equilibrium. We also established a quantitative relationship 
between the equilibria of the regularized and unregularized games, providing insight 
into equilibrium selection. We then analyzed learning dynamics for computing equilibria. While gradient-based 
and best-response methods converge in smooth settings,  the former shows slower convergence in non smooth games, whereas the latter may not converge to  the 
potential-maximizing equilibrium. To overcome these limitations, 
we proposed an MM-based learning scheme, which guarantees convergence to a 
potential-optimal equilibrium and demonstrates faster convergence behavior. 
These results highlight the importance of regularization and optimization-based 
learning in achieving robust equilibrium selection, and suggest several directions 
for future work, including extensions to non-convex and more general behavioral models. While we have shown convergence of the MM based learning scheme for a particular form of the minorizing function, extending these results to other forms of the minorizer remains open. Carefully choosing the minorizer can speed up the convergence, as well as ameliorate technical difficulties arising from non-smoothness or non-concavity of the utility.

\subsection*{Declaration of Generative AI and AI-Assisted Technologies}
The authors used Claude (Opus 5.5, Anthropic) to assist in checking and deriving intermediate algebraic steps and formatting mathematical derivations in Sections \ref{sec:convergence}. The authors have independently re-derived and verified all the content of the paper, and take full responsibility for the mathematical accuracy and content of the publication.
\bibliographystyle{ieeetr}

\bibliography{cas-refs}

@book{von1944theory,
  title={Theory of games and economic behavior.},
  author={Von Neumann, J and Morgenstern, O},
  year={1944},
  publisher={Princeton University Press}
}

@article{rosen1965existence,
  title={Existence and uniqueness of equilibrium points for concave n-person games},
  author={Rosen, J Ben},
  journal={Econometrica: Journal of the Econometric Society},
  pages={520--534},
  year={1965},
  publisher={JSTOR}
}

@book{jallais2005allais,
  title={The {A}llais paradox and its immediate consequences for expected utility theory},
  author={Jallais, Sophie and Pradier, Pierre-Charles},
  journal={The experiment in the history of economics},
  pages={25--49},
  year={2005},
  publisher={Routledge New York}
}

@article{kai1979prospect,
  title={Prospect theory: An analysis of decision under risk},
  author={Kahneman, Daniel and Tversky, Amos},
  journal={Econometrica},
  volume={47},
  number={2},
  pages={363--391},
  year={1979}
}

@book{bach2013learning,
  title={Learning with Submodular Functions: A Convex Optimization Perspective},
  author={Bach, Francis},
  year={2013},
  publisher={Foundations and Trends in Machine Learning}
}

@article{stott2006cumulative,
  title={Cumulative prospect theory's functional menagerie},
  author={Stott, Henry P},
  journal={Journal of Risk and uncertainty},
  volume={32},
  pages={101--130},
  year={2006},
  publisher={Springer}
}

@article{holmes2011management,
  title={Management theory applications of prospect theory: Accomplishments, challenges, and opportunities},
  author={Holmes Jr, R Michael and Bromiley, Philip and Devers, Cynthia E and Holcomb, Tim R and McGuire, Jean B},
  journal={Journal of Management},
  volume={37},
  number={4},
  pages={1069--1107},
  year={2011},
  publisher={Sage Publications Sage CA: Los Angeles, CA}
}

@article{metzger2019non,
  title={Non-cooperative games with prospect theory players and dominated strategies},
  author={Metzger, Lars Peter and Rieger, Marc Oliver},
  journal={Games and Economic Behavior},
  volume={115},
  pages={396--409},
  year={2019},
  publisher={Elsevier}
}

@article{keskin2016equilibrium,
  title={Equilibrium notions for agents with cumulative prospect theory preferences},
  author={Keskin, Kerim},
  journal={Decision Analysis},
  volume={13},
  number={3},
  pages={192--208},
  year={2016},
  publisher={INFORMS}
}

@article{vahid2019modeling,
  title={Modeling noncooperative game of GENCOs’ participation in electricity markets with prospect theory},
  author={Vahid-Pakdel, MJ and Ghaemi, Sina and Mohammadi-Ivatloo, Behnam and Salehi, Javad and Siano, Pierluigi},
  journal={IEEE Transactions on Industrial Informatics},
  volume={15},
  number={10},
  pages={5489--5496},
  year={2019},
  publisher={IEEE}
}

@article{merrick2016modeling,
  title={Modeling adversaries in counterterrorism decisions using prospect theory},
  author={Merrick, Jason RW and Leclerc, Philip},
  journal={Risk Analysis},
  volume={36},
  number={4},
  pages={681--693},
  year={2016},
  publisher={Wiley Online Library}
}

@article{shalev2000loss,
  title={Loss aversion equilibrium},
  author={Shalev, Jonathan},
  journal={International Journal of Game Theory},
  volume={29},
  pages={269--287},
  year={2000},
  publisher={Springer}
}

@article{gan2022application,
  title={Application and outlook of prospect theory applied to bounded rational power system economic decisions},
  author={Gan, Lei and Hu, Yangyi and Chen, Xingying and Li, Gengyin and Yu, Kun},
  journal={IEEE Transactions on Industry Applications},
  volume={58},
  number={3},
  pages={3227--3237},
  year={2022},
  publisher={IEEE}
}

@article{fochesato2025noncooperative,
  title={Noncooperative games with prospect theoretic preferences},
  author={Fochesato, Marta and Pokou, Fr{\'e}dy and Le Cadre, H{\'e}l{\`e}ne and Lygeros, John},
  journal={IEEE Control Systems Letters},
  year={2025},
  publisher={IEEE}
}

@article{monderer1996potential,
  title={Potential games},
  author={Monderer, Dov and Shapley, Lloyd S},
  journal={Games and economic behavior},
  volume={14},
  number={1},
  pages={124--143},
  year={1996},
  publisher={Elsevier}
}

@article{jensen2010aggregative,
  title={Aggregative games and best-reply potentials},
  author={Jensen, Martin Kaae},
  journal={Economic theory},
  volume={43},
  number={1},
  pages={45--66},
  year={2010},
  publisher={Springer}
}

@article{neyman1997correlated,
  title={Correlated equilibrium and potential games},
  author={Neyman, Abraham},
  journal={International Journal of Game Theory},
  volume={26},
  number={2},
  pages={223--227},
  year={1997},
  publisher={Springer}
}

@article{lin2007multi,
  title={The multi-agent rendezvous problem. Part 1: The synchronous case},
  author={Lin, Jie and Morse, A Stephen and Anderson, Brian DO},
  journal={SIAM Journal on Control and Optimization},
  volume={46},
  number={6},
  pages={2096--2119},
  year={2007},
  publisher={SIAM}
}

@article{cortes2017coordinated,
  title={Coordinated control of multi-robot systems: A survey},
  author={Cort{\'e}s, Jorge and Egerstedt, Magnus},
  journal={SICE Journal of Control, Measurement, and System Integration},
  volume={10},
  number={6},
  pages={495--503},
  year={2017},
  publisher={Taylor \& Francis}
}

@article{lange2021nonconvex,
  title={Nonconvex optimization via {MM} algorithms: Convergence theory},
  author={Lange, Kenneth and Won, Joong-Ho and Landeros, Alfonso and Zhou, Hua},
  journal={arXiv preprint arXiv:2106.02805},
  year={2021}
}

@book{bertsekas1999nonlinear,
  title={Nonlinear programming},
  author={Bertsekas, Dimitri P},
  publisher={Athena Scientific},
  year={1999}
}

@inproceedings{ashok2025achieving,
  title={Achieving a {C}ollective {T}arget Through {I}ncentives},
  author={Krishnan K.S., Ashok and Le Cadre, H{\'e}l{\`e}ne and Bu{\v{s}}i{\'c}, Ana},
  booktitle={International Conference on Network Games, Artificial Intelligence, Control and Optimization},
  pages={57--67},
  year={2025},
  organization={Springer}
}

@inproceedings{ks2025irrationality,
  title={How irrationality shapes {N}ash equilibria: A prospect-theoretic perspective},
  author={Krishnan K.S., Ashok  and Le Cadre, H{\'e}l{\`e}ne and Bu{\v{s}}i{\'c}, Ana},
  booktitle={2025 IEEE 64th Conference on Decision and Control (CDC)},
  pages={4428--4433},
  year={2025},
  organization={IEEE}
}

@article{gilles2023emergent,
  title={Emergent collaboration in social purpose games},
  author={Gilles, Robert P and Mallozzi, Lina and Messalli, Roberta},
  journal={Dynamic Games and Applications},
  volume={13},
  number={2},
  pages={566--588},
  year={2023},
  publisher={Springer}
}

@book{voorneveld1999potential,
  title={Potential games and interactive decisions with multiple criteria},
  author={Voorneveld, Mark},
  publisher={Ph.D. Thesis, Tilburg University},
  year={1999}
}

@article{mckelvey1995quantal,
  title={Quantal response equilibria for normal form games},
  author={McKelvey, Richard D and Palfrey, Thomas R},
  journal={Games and economic behavior},
  volume={10},
  number={1},
  pages={6--38},
  year={1995},
  publisher={Elsevier}
}

@article{alos2010logit,
  title={The logit-response dynamics},
  author={Al{\'o}s-Ferrer, Carlos and Netzer, Nick},
  journal={Games and Economic Behavior},
  volume={68},
  number={2},
  pages={413--427},
  year={2010},
  publisher={Elsevier}
}

@article{aghassi2006robust,
  title={Robust game theory},
  author={Aghassi, Michele and Bertsimas, Dimitris},
  journal={Mathematical programming},
  volume={107},
  number={1},
  pages={231--273},
  year={2006},
  publisher={Springer}
}

@book{cooper1999coordination,
  title={Coordination games},
  author={Cooper, Russell},
  year={1999},
  publisher={{C}ambridge university Press}
}

@article{lasry2007mean,
  title={Mean field games},
  author={Lasry, Jean-Michel and Lions, Pierre-Louis},
  journal={Japanese journal of mathematics},
  volume={2},
  number={1},
  pages={229--260},
  year={2007},
  publisher={Springer}
}

@article{mccain2008cooperative,
  title={Cooperative games and cooperative organizations},
  author={McCain, Roger A},
  journal={The Journal of Socio-Economics},
  volume={37},
  number={6},
  pages={2155--2167},
  year={2008},
  publisher={Elsevier}
}

@article{kukushkin2004best,
  title={Best response dynamics in finite games with additive aggregation},
  author={Kukushkin, Nikolai S},
  journal={Games and Economic Behavior},
  volume={48},
  number={1},
  pages={94--110},
  year={2004},
  publisher={Elsevier}
}

@article{hunter2004tutorial,
  title={A tutorial on MM algorithms},
  author={Hunter, David R and Lange, Kenneth},
  journal={The American Statistician},
  volume={58},
  number={1},
  pages={30--37},
  year={2004},
  publisher={Taylor \& Francis}
}

@article{parikh2014proximal,
  title={Proximal algorithms},
  author={Parikh, Neal and Boyd, Stephen},
  journal={Foundations and Trends in optimization},
  volume={1},
  number={3},
  pages={127--239},
  year={2014},
  publisher={Emerald Publishing Limited}
}

@article{rockafellar1976monotone,
  title={Monotone operators and the proximal point algorithm},
  author={Rockafellar, R Tyrrell},
  journal={SIAM journal on control and optimization},
  volume={14},
  number={5},
  pages={877--898},
  year={1976},
  publisher={SIAM}
}

@inproceedings{candogan2011learning,
  title={Learning in near-potential games},
  author={Candogan, Ozan and Ozdaglar, Asuman and Parrilo, Pablo A},
  booktitle={2011 50th IEEE Conference on Decision and Control and European Control Conference},
  pages={2428--2433},
  year={2011},
  organization={IEEE}
}

@article{hao2023exploration,
  title={Exploration in deep reinforcement learning: From single-agent to multiagent domain},
  author={Hao, Jianye and Yang, Tianpei and Tang, Hongyao and Bai, Chenjia and Liu, Jinyi and Meng, Zhaopeng and Liu, Peng and Wang, Zhen},
  journal={IEEE transactions on neural networks and learning systems},
  volume={35},
  number={7},
  pages={8762--8782},
  year={2023},
  publisher={IEEE}
}

@article{zhang2023routing,
  title={Routing optimization with vehicle--customer coordination},
  author={Zhang, Wei and Jacquillat, Alexandre and Wang, Kai and Wang, Shuaian},
  journal={Management Science},
  volume={69},
  number={11},
  pages={6876--6897},
  year={2023},
  publisher={INFORMS}
}

@book{facchinei2003finite,
  title={Finite-dimensional variational inequalities and complementarity problems},
  author={Facchinei, Francisco and Pang, Jong-Shi},
  year={2003},
  publisher={Springer}
}

@article{milgrom1990rationalizability,
  title={Rationalizability, learning, and equilibrium in games with strategic complementarities},
  author={Milgrom, Paul and Roberts, John},
  journal={Econometrica: Journal of the Econometric Society},
  pages={1255--1277},
  year={1990},
  publisher={JSTOR}
}

@article{glicksberg1952further,
  title={A further generalization of the Kakutani fixed theorem, with application to Nash equilibrium points},
  author={Glicksberg, Irving L},
  journal={Proceedings of the American Mathematical Society},
  volume={3},
  number={1},
  pages={170--174},
  year={1952}
}

@incollection{jensen2018aggregative,
  title={Aggregative games},
  author={Jensen, Martin Kaae},
  booktitle={Handbook of Game Theory and Industrial Organization, Volume I},
  pages={66--92},
  year={2018},
  publisher={Edward Elgar Publishing}
}

@article{dindovs2006better,
  title={Better-reply dynamics and global convergence to Nash equilibrium in aggregative games},
  author={Dindo{\v{s}}, Martin and Mezzetti, Claudio},
  journal={Games and Economic Behavior},
  volume={54},
  number={2},
  pages={261--292},
  year={2006},
  publisher={Elsevier}
}

@inproceedings{even2009convergence,
  title={On the convergence of regret minimization dynamics in concave games},
  author={Even-Dar, Eyal and Mansour, Yishay and Nadav, Uri},
  booktitle={Proceedings of the forty-first annual ACM symposium on Theory of computing},
  pages={523--532},
  year={2009}
}

@article{hart2000simple,
  title={A simple adaptive procedure leading to correlated equilibrium},
  author={Hart, Sergiu and Mas-Colell, Andreu},
  journal={Econometrica},
  volume={68},
  number={5},
  pages={1127--1150},
  year={2000},
  publisher={Wiley Online Library}
}

@article{bauschke2015fej,
  title={On Fej\'er monotone sequences and nonexpansive mappings},
  author={Bauschke, Heinz H and Dao, Minh N and Moursi, Walaa M},
  journal={preprint arXiv:1507.05585},
  year={2015}
  }

@article{bubeck2015convex,
  title={Convex optimization: Algorithms and complexity},
  author={Bubeck, S{\'e}bastien},
  journal={Foundations and trends in Machine Learning},
  volume={8},
  number={3-4},
  pages={231--357},
  year={2015},
  publisher={Emerald Publishing limited}
}

@book{bertsekas2009convex,
  title={Convex optimization theory},
  author={Bertsekas, Dimitri P},
  year={2009},
  publisher={Athena Scientific Belmont}
}

@book{rockafellar1970convex,
  title={Convex analysis},
  author={Rockafellar, R Tyrrell},
  year={1970},
  publisher={Princeton university press}
}

@book{basar1998dynamic,
  title={Dynamic noncooperative game theory},
  author={Ba{\c{s}}ar, Tamer and Olsder, Geert Jan},
  year={1998},
  publisher={SIAM}
}

@book{gibbons1992game,
  title={Game theory for applied economists},
  author={Gibbons, Robert S},
  year={1992},
  publisher={Princeton University Press}
}

@article{rosenthal1973class,
  title={A class of games possessing pure-strategy Nash equilibria},
  author={Rosenthal, Robert W},
  journal={International journal of game theory},
  volume={2},
  number={1},
  pages={65--67},
  year={1973},
  publisher={Physica-Verlag Heidelberg}
}

@inproceedings{chen2008altruism,
  title={Altruism, selfishness, and spite in traffic routing},
  author={Chen, Po-An and Kempe, David},
  booktitle={Proceedings of the 9th ACM Conference on Electronic Commerce},
  pages={140--149},
  year={2008}
}

@article{ma2011decentralized,
  title={Decentralized charging control of large populations of plug-in electric vehicles},
  author={Ma, Zhongjing and Callaway, Duncan S and Hiskens, Ian A},
  journal={IEEE Transactions on control systems technology},
  volume={21},
  number={1},
  pages={67--78},
  year={2011},
  publisher={IEEE}
}

@article{bramoulle2007public,
  title={Public goods in networks},
  author={Bramoull{\'e}, Yann and Kranton, Rachel},
  journal={Journal of Economic theory},
  volume={135},
  number={1},
  pages={478--494},
  year={2007},
  publisher={Elsevier}
}

@article{holmstrom1982moral,
  title={Moral hazard in teams},
  author={Holmstrom, Bengt},
  journal={The Bell journal of economics},
  pages={324--340},
  year={1982},
  publisher={JSTOR}
}

@article{tversky1992advances,
  title={Advances in prospect theory: Cumulative representation of uncertainty},
  author={Tversky, Amos and Kahneman, Daniel},
  journal={Journal of Risk and uncertainty},
  volume={5},
  number={4},
  pages={297--323},
  year={1992},
  publisher={Springer}
}



\end{document}